\documentclass[11pt]{article}

\usepackage[T1]{fontenc}
\usepackage[utf8]{inputenc}
\usepackage{authblk}
\usepackage[authoryear]{natbib}
\usepackage{geometry}
\usepackage{ifthen}
\usepackage{fancybox}
\usepackage{amsmath,amssymb,amsthm}
\usepackage[ruled,vlined,commentsnumbered,titlenotnumbered]{algorithm2e}
\usepackage{algorithmic}
\usepackage{tikz}
\usetikzlibrary{chains,arrows}
\usepackage{dsfont}
\usepackage{todonotes}
\usepackage{mathtools}
\usepackage{stmaryrd}
\usepackage{color-edits}
\usepackage{xcolor}
\usepackage{hhline}
\usepackage{caption}
\usepackage{subcaption}
\definecolor{cornellred}{rgb}{0.7, 0.11, 0.11}
\definecolor{dgreen}{rgb}{0.0, 0.5, 0.0}
\definecolor{ballblue}{rgb}{0.13, 0.67, 0.8}
\definecolor{bleudefrance}{rgb}{0.19, 0.55, 0.91}
\usepackage{hyperref}
\hypersetup{
    colorlinks = true,
    linkcolor={blue},
    citecolor=cornellred,
    linkbordercolor = {white}
}

\usepackage{cleveref}
\usepackage[authoryear]{natbib}
\usepackage{xspace}
\usepackage{physics}
\usepackage{graphicx}
\usepackage[scaled=0.85]{beramono}
\usepackage{capt-of}
\usepackage{csquotes}
\addauthor{rn}{red}     
\addauthor{jw}{blue}    
\addauthor{tr}{green}   
\theoremstyle{plain}
\newtheorem{proposition}{Proposition}
\newtheorem{theorem}{Theorem}
\newtheorem{lemma}{Lemma}

\theoremstyle{definition}
\newtheorem{definition}{Definition}

\newtheorem{remark}{Remark}

\newcommand{\ie}{i.e.,\xspace}

\newcommand{\RR}{\ensuremath{\mathbb{R}}}

\newcommand{\Ex}[1]{\mbox{\rm\bf E}\left[#1\right]}

\DeclareMathOperator{\argmax}{argmax}
\DeclareMathOperator{\argmin}{argmin}

\newcommand{\dvx}[2]{\frac{\partial #1}{\partial #2}}
\newcommand{\dvxx}[2]{\frac{\partial^2 #1}{\partial {#2}^2 }}
\newcommand{\dvxy}[3]{\frac{\partial^2 #1}{\partial {#2}\partial {#3}}}
\newcommand{\F}{\mathcal{F}}
\newcommand{\xbf}{\textbf{x}}

\newcommand{\xv}{\textbf{x}}
\newcommand{\yv}{\textbf{y}}
\newcommand{\rv}{\textbf{r}}
\newcommand{\cen}[1]{\texttt{conc-env}(#1)}
\newcommand{\zu}{{Z}_{u}}
\newcommand{\zl}{{Z}_{l}}
\newcommand{\za}{\hat{z}}
\newcommand{\Xbf}{\mathbf{X}}
\newcommand{\Ybf}{\mathbf{Y}}
\newcommand{\Pbf}{\mathcal{P}}

\newcommand{\xo}{x^*}
\newcommand{\xob}{\mathbf{\xo}}
\newcommand{\Obf}{\mathbf{O}}
\newcommand{\val}{\mathcal{V}^{(i)}}
\def\noqed{\renewcommand{\qedsymbol}{}}

\newcommand{\Hbf}{\mathbf{H}}
\newcommand{\hbf}{\mathbf{h}}

\newcommand{\NRWGameAlg}{\texttt{GAME}}
\newcommand{\NRWBinarySearchAlg}{\texttt{BINARY}}
\newcommand{\BMBKAlg}{\texttt{BMBK}}
\newcommand{\Lbf}{\mathbf{L}}
\newcommand{\Ibf}{\mathbf{I}}
\newcommand{\Dbf}{\mathbf{D}}
\newcommand{\Vbf}{\mathbf{V}}

\newcommand\join{\lor}
\newcommand\meet{\land}
\renewcommand\L{\mathcal{L}}

\begin{document}

\setcounter{page}{1} 

\title{Optimal Algorithms for Continuous
  Non-monotone Submodular and DR-Submodular Maximization}

\author[$\dagger$]{Rad Niazadeh}
\author[$\dagger$]{Tim Roughgarden}
\author[$\dagger$]{Joshua R. Wang}
\affil[$\dagger$]{Department of Computer Science, Stanford University}
\renewcommand\Authands{ and }
\date{}

\maketitle

\begin{abstract}

In this paper we study the fundamental problems of maximizing a continuous non-monotone submodular function over the hypercube, both with and without coordinate-wise concavity. This family of optimization problems has several applications in machine learning, economics, and communication systems. Our main result is the first $\frac{1}{2}$-approximation algorithm for continuous submodular
function maximization; this approximation factor of~$\tfrac{1}{2}$ is the best
possible for algorithms that only query the objective function at polynomially many points.  For
the special case of DR-submodular maximization, i.e. when the submodular functions is also coordinate-wise concave along all coordinates, we provide a different
$\frac{1}{2}$-approximation algorithm that runs in quasi-linear time.
Both of these results improve upon prior work~\citep{bian2017continuous,bian2017guaranteed,soma2017non}. 
  
  Our first
  algorithm uses novel ideas such as reducing the guaranteed
  approximation problem to analyzing a zero-sum game for each
  coordinate, and incorporates the geometry of this zero-sum game to
  fix the value at this coordinate. Our second algorithm 
exploits 
coordinate-wise concavity to identify a monotone equilibrium
  condition sufficient for getting the required approximation
  guarantee, and hunts for the equilibrium point using binary
  search. We further run experiments to verify the performance of our
  proposed algorithms in related machine learning applications. 

\end{abstract}
\section{Introduction}


{\em Submodular optimization} is a sweet spot between tractability and
expressiveness, with numerous applications in machine learning
(e.g. \cite{krause2014submodular}, and see below) while permitting many algorithms that are
both practical and backed by rigorous guarantees (e.g. \cite{buchbinder2015tight,feige2011maximizing,calinescu2011maximizing}). In general, a real-valued function~$\F$ defined on a lattice $\L$
is {\em submodular} if and only if
\[
\F(x \join y) + \F(x \meet y) \le \F(x) + \F(y)
\]
for all $x,y \in \L$, where $x \join y$ and $x \meet y$ denote the
join and meet, respectively, of $x$ and $y$ in the lattice $\L$.  
Such functions are generally neither convex nor concave.
In one of the most commonly studied examples, $\L$ is the lattice of
subsets of a fixed ground set (or a sublattice thereof), with union
and intersection playing the roles of join and meet, respectively.

This paper concerns a different well-studied setting, where $\L$ is a
hypercube (i.e., $[0,1]^n$), with componentwise maximum and minimum
serving as the join and meet, respectively.\footnote{Our results also
  extend easily to arbitrary axis-aligned boxes (i.e., ``box
  constraints'').}  We consider the fundamental problem of
(approximately) maximizing a continuous and nonnegative submodular
function over the hypercube.\footnote{More generally, the function
  only has to be nonnegative at the points $\vec{0}$ and $\vec{1}$.}
The function~$\F$ is given as a ``black
box'': accessible only via querying its value at a point.
We are interested in algorithms that use at most a polynomial (in~$n$)
number of queries.  We do not assume that~$\F$ is monotone (otherwise
the problem is trivial).


We next briefly mention four applications of maximizing a non-monotone
submodular function over a hypercube that are germane to machine
learning and other related application domains.\footnote{See the supplement for more details on these applications.}

\noindent
\textbf{Non-concave quadratic programming.}
In this problem, the goal
is to maximize $\F(\xbf)=\frac{1}{2}\xbf^T\Hbf\xbf+\hbf^T\xbf+c$,
where the off-diagonal entries of $\Hbf$ are non-positive.
One application of this problem is to large-scale
price optimization on the basis of demand forecasting models~\citep{ito2016large}.


\noindent
\textbf{Map inference for Determinantal Point Processes (DPP).}
  DPPs are elegant probabilistic models that arise in statistical
  physics and random matrix theory.  DPPs can be used as generative
  models in applications such as text summarization, human pose
  estimation, and news threading
  tasks~\citep{kulesza2012determinantal}.  
The approach in~\cite{gillenwater2012near} to the problem boils down
to maximize a suitable submodular function over the hypercube, accompanied with an appropriate rounding (see
also~\citep{bian2017continuous}). One can also think of regularizing this objective function with $\ell_2\textrm{-norm}$ regularizer, in order to avoid overfitting. Even with a regularizer, the function remains submodular.

\noindent
\textbf{Log-submodularity and mean-field inference.}
Another probabilistic model that generalizes DPPs and all other
  strong Rayleigh measures~\citep{li2016fast,zhang2015higher} is the
  class of \emph{log-submodular} distributions over sets,
  i.e. $p(S)\sim\exp(\F(S))$ where $\F(\cdot)$ is a set submodular
  function. MAP inference over this distribution has applications in
  machine learning~\citep{djolonga2014map}.
  One variational approach towards this MAP inference task is to use
  \emph{mean-field inference} to approximate the distribution $p$ with
  a product distribution $\xbf\in [0,1]^n$, which again boils down to
  submodular function maximization over the hypercube 
(see~\citep{bian2017continuous}).

\noindent
\textbf{Revenue maximization over social networks.} In this problem, there
  is a seller who wants to sell a product over a social network of
  buyers. To do so, the seller gives away trial products and fractions
  thereof to the buyers in the
  network~\citep{bian2017guaranteed,hartline2008optimal}.
  In~\citep{bian2017guaranteed}, there is 
an objective function that takes into account two
  parts: the revenue gain from those who did not get a free product,
  where the revenue function for any such buyer is a non-negative
  non-decreasing and submodular function $R_i(\xbf)$; and the revenue
  loss from those who received the free product, where the revenue
  function for any such buyer is a non-positive non-increasing and
  submodular function $\bar{R}_i(\xbf)$.  The combination for all
  buyers is a non-monotone submodular function. It also is non-negative
  at $\vec{0}$ and $\vec{1}$, by extending the model and accounting for extra revenue gains from buyers with free trials.




\paragraph{Our results.}
Maximizing a submodular function over the hypercube is at least as
difficult as over the subsets of a ground set.\footnote{An instance of
  the latter problem can be converted to one of the former by
  extending the given set function~$f$ (with domain viewed as
  $\{0,1\}^n$) to its multilinear extension~$\F$ defined on the
  hypercube (where $\F(\xbf) = \sum_{S \subseteq [n]} \prod_{i \in
    S} x_i \prod_{i \notin S} (1-x_i) f(S)$).  Sampling based on an $\alpha$-approximate
  solution for the multilinear extension yields an equally
  good approximate solution to the original problem.}
For the latter problem, the best approximation ratio achievable by an
algorithm making a polynomial number of queries is $\tfrac{1}{2}$; 
the (information-theoretic) lower bound is due to~\citep{feige2011maximizing}, the
optimal algorithm to~\citep{buchbinder2015tight}.  Thus, the best-case scenario for
maximizing a submodular function over the hypercube (using
polynomially many queries) is a $\tfrac{1}{2}$-approximation.  The
main result of this paper achieves this best-case scenario:
\newpage
\begin{displayquote}
\emph{There is an algorithm for maximizing a continuous submodular function over the
hypercube that guarantees a $\tfrac{1}{2}$-approximation while using
only a polynomial number of queries to the function under mild continuity assumptions.}
\end{displayquote}
Our algorithm is inspired by the \emph{bi-greedy} algorithm of~\cite{buchbinder2015tight}, which maximizes a submodular set function; it maintains two solutions initialized at
$\vec{0}$ and $\vec{1}$, go over coordinates sequentially, and make
the two solutions agree on each coordinate. The algorithmic question here is
how to choose the new coordinate value for the two solutions, so that the algorithm gains enough value relative to the optimum in each
iteration. Prior to our work, the best-known result was a $\tfrac{1}{3}$-approximation~\citep{bian2017guaranteed}, which is also inspired by the bi-greedy. Our algorithm requires a number of new ideas, including a reduction to the analysis of a zero-sum game for each coordinate, and the use of the special geometry of this game to bound the value of the game.

The second and third applications above induce objective functions
that, in addition to being submodular, are concave in each coordinate\footnote{However, after regularzation the function still remains submodular, but can lose coordinate-wise concavity.}
(called {\em DR-submodular} in~\citep{soma2015generalization} based on diminishing returns defined in \citep{kapralov2013online}).  Here, an optimal
$\tfrac{1}{2}$-approximation algorithm was recently already known on integer lattices~\citep{soma2017non}, that can easily be generalized to our continuous setting as well;
our contribution is a significantly faster such bi-greedy algorithm.  The main
idea here is to identify a monotone equilibrium condition sufficient
for getting the required approximation guarantee, which enables a
binary search-type solution.

We also run experiments to verify the performance of our proposed
algorithms in practical machine learning applications. We observe that
our algorithms match the performance of the prior work, while
providing either a better guaranteed approximation or a better running
time.







\paragraph{Further related work.} 
\citet{buchbinder2016deterministic} derandomize the bi-greedy algorithm.
\citet{staib2017robust} apply continuous submodular optimization to budget allocation, and develop a new submodular optimization algorithm to this end.
\citet{hassani2017gradient} give a $\frac12$-approximation for \emph{monotone} continuous submodular functions under convex constraints.
\citet{gotovos2015non} consider (adaptive) submodular maximization when feedback is given after an element is chosen.
\citet{chen2018online, roughgarden2018optimal} consider submodular maximization in the context of online no-regret learning.
\cite{mirzasoleiman2013distributed} show how to perform submodular maximization with distributed computation. Submodular minimization has been studied in \cite{schrijver2000combinatorial,iwata2001combinatorial}. See \cite{bach2013learning} for a survey on more applications in machine learning.

\paragraph{{Variations of continuous submodularity.}}
We consider non-monotone non-negative \emph{continuous submodular functions}, i.e. $\F:[0,1]\rightarrow [0,1]^n$ s.t. $\forall \xv,\yv\in[0,1]^n$, $\F(\xv)+\F(\yv)\geq \F(\xv\vee\yv)+\F(\xv\wedge\yv)$, where $\vee$ and $\wedge$ are coordinate-wise max and min operations. Two related properties are \emph{weak Diminishing Returns Submodularity} (\texttt{weak DR-SM}) and \emph{strong Diminishing Returns Submodularity} (\texttt{strong DR-SM})~\citep{bian2017guaranteed}, formally defined below. Indeed, \texttt{weak DR-SM} is equivalent to submodularity (see \Cref{prop:weak-strong-multi-defs} in the supplement), and hence we use these terms interchangeably.
\begin{definition}[\texttt{Weak/Strong DR-SM}] 
\label{def:weak-strong-DR}
Consider a continuous function $\F:[0,1]^n\rightarrow [0,1]$:
\begin{itemize}
\leftskip-0.9cm
\item \texttt{Weak DR-SM} (continuous submodular): $\forall i\in[n],~\forall \xv_{-i}\leq \yv_{-i}\in[0,1]^n$, and $\forall \delta\geq 0, \forall z$
\[
\F(z+\delta,\xv_{-i} )-\F(z,\xv_{-i})\geq \F(z+\delta,\yv_{-i} )-\F(z,\yv_{-i})
\]
\item \texttt{Strong DR-SM} (\emph{DR-submodular}
): $\forall i\in[n],~\forall\xv\leq \yv\in[0,1]^n$, and $\forall \delta\geq 0$:
\[
\F(x_i+\delta,\xv_{-i} )-\F(\xv)\geq \F(y_i+\delta,\yv_{-i} )-\F(\yv)
\]
\end{itemize}
\end{definition}
As simple corollaries, a twice-differentiable $\F$ is \texttt{strong DR-SM}
 if and only if all the entries of its Hessian
are non-positive, and \texttt{weak DR-SM} if and only if all of the \emph{off-diagonal} entries of
its Hessian are non-positive. Also, \texttt{weak DR-SM} together with
concavity along each coordinate is equivalent to \texttt{strong DR-SM}
(see \Cref{prop:weak-strong-multi-defs} in the supplementary materials
for more details).
\paragraph{{Coordinate-wise Lipschitz continuity.}}
Consider univariate functions generated by fixing all but one of the coordinates of the original function $\F(\cdot)$. In future sections, we sometimes require mild technical assumptions on the Lipschitz continuity of these single dimensional functions.
\begin{definition}[Coordinate-wise Lipschitz] 
\label{def:coordinate-Lip}

A function $\F:[0,1]^n\rightarrow [0,1]$ is \emph{coordinate-wise Lipschitz continuous} if there exists a constant $C>0$ such that $\forall i\in[n]$, $\forall \xv_{-i}\in[0,1]^n$, the single variate function $\F(\cdot,\xv_{-i})$ is $C$-Lipschitz continuous, \ie
\begin{equation*}
\forall z_1,z_2\in[0,1]:~~\lvert \F(z_1,\xv_{-i})-\F(z_2,\xv_{-i}) \rvert \leq C\lvert z_1-z_2 \rvert
\end{equation*}
\end{definition}
%
\vspace{-5mm}
\section{Weak DR-SM Maximization: Continuous Randomized Bi-Greedy}
\label{sec:weak-DR}
Our first main result is a 
$\frac{1}{2}$-approximation algorithm (up to additive
error $\delta$) for maximizing a continuous submodular function $\F$,
a.k.a. \texttt{weak DR-SM}, which is information-theoretically 
optimal~\citep{feige2011maximizing}. 
This result assumes that $\F$ is coordinate-wise Lipschitz
continuous.\footnote{Such an assumption is necessary, since otherwise the
  single-dimensional problem amounts to optimizing an arbitrary function and
  is hence intractable. Prior work,
  e.g.~\cite{bian2017guaranteed} and \cite{bian2017continuous},
  implicitly requires such an assumption to perform single-dimensional optimization.} Before describing our algorithm, we introduce the notion of the \emph{positive-orthant concave envelope} of a two-dimensional curve, which is useful for understanding our algorithm.

\begin{definition} 
\label{def:concave-envelope} 
Consider a curve $\rv(z)=(g(z),h(z))\in\mathbb{R}^2$ over the interval $z \in [\zl,\zu]$ such that:
\begin{enumerate}
\item $g:[\zl,\zu]\rightarrow [-1,\alpha]$ and $h: [\zl,\zu]\rightarrow [-1,\beta]$ are both continuous,
\item $g(\zl)=h(\zu)=0$, and  $h(\zl)=\beta\in [0,1],~g(\zu)=\alpha\in[0,1]$.
\end{enumerate}
Then the \emph{positive-orthant concave envelope} of $\rv(\cdot)$, denoted by $\cen{\rv}$, is the smallest concave curve in the positive-orthant upper-bounding all the points $\{\mathbf{r}(z):z\in[\zl,\zu]\}$ (see \Cref{fig:gh-curve}), \ie
\begin{equation*}
\cen{\rv}\triangleq \textrm{upper-face}\left(\textrm{conv}\left(\{\mathbf{r}(z):z\in[\zl,\zu]\}\right)\cap \left\{(g',h')\in[0,1]^2:\frac{h'}{\beta}+ \frac{g'}{\alpha}\geq 1\right\}\right)
\end{equation*}
\end{definition}
\begin{figure}
        \centering
        \begin{subfigure}{0.8\textwidth}
            
                 \includegraphics[scale=0.12]{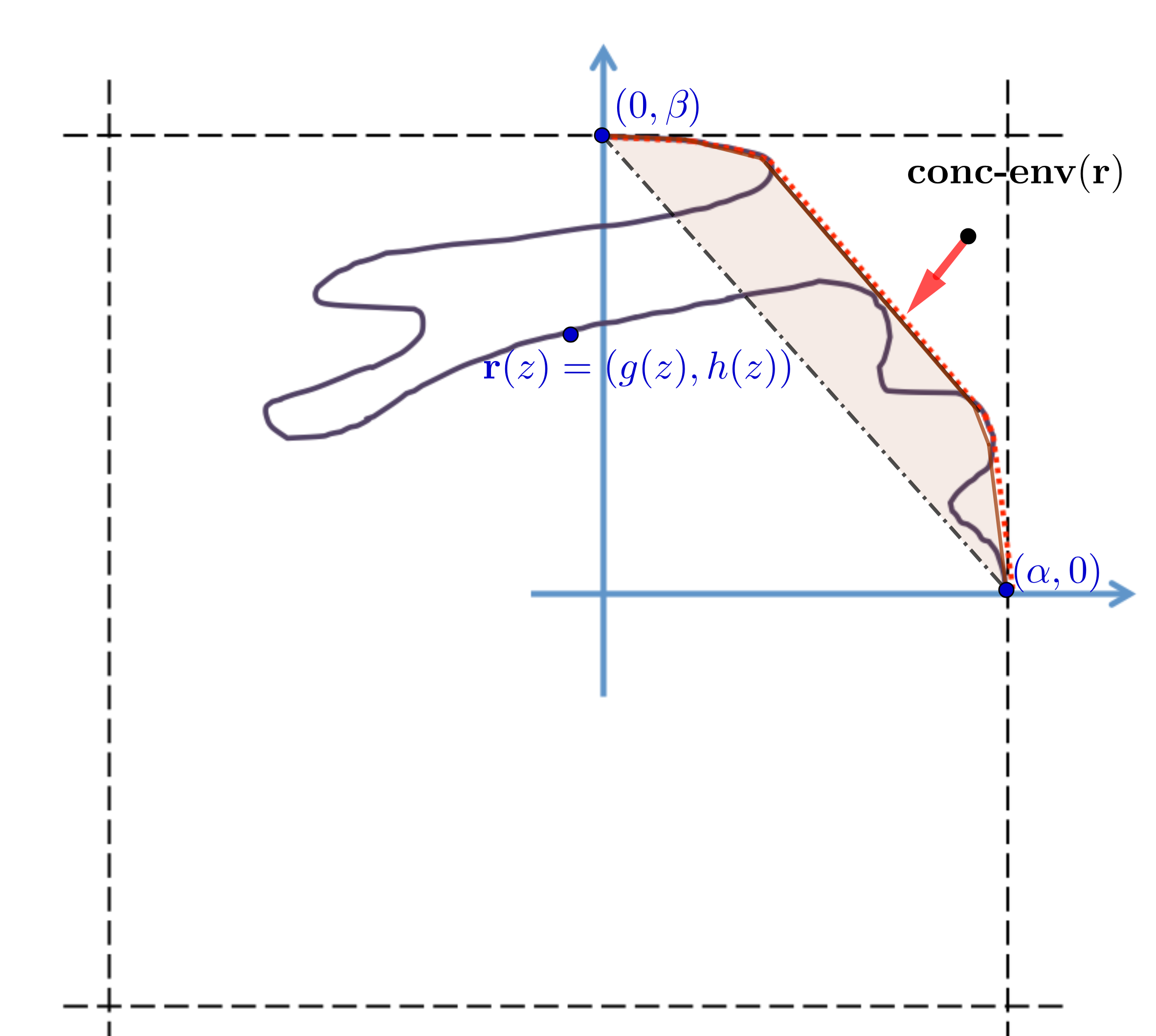}
\caption{ \label{fig:gh-curve} Continuous curve $\mathbf{r}(z)$ in $\mathbb{R}^2$ (dark blue), positive-orthant concave envelope (red).} 
        \end{subfigure}\\
        \begin{subfigure}{0.75\textwidth}
                \centering
                \leftskip=-28mm
                \includegraphics[scale=0.65]{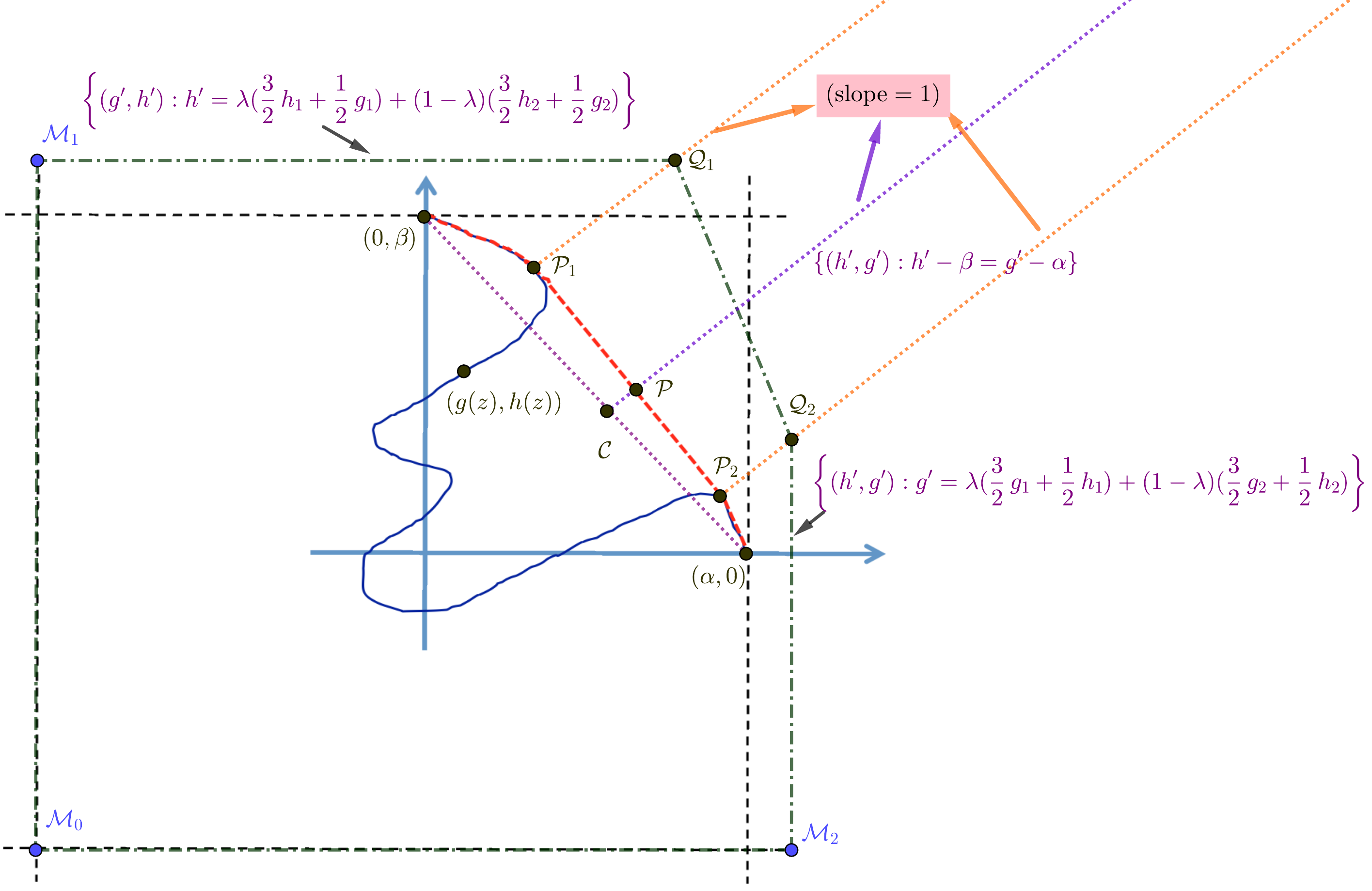}
\caption{\label{fig:pos-region} Pentagon $(\mathcal{M}_0,\mathcal{M}_1,\mathcal{Q}_1,\mathcal{Q}_2,\mathcal{M}_2)$= \texttt{ADV} player's positive region against a mixed strategyover two points $\Pbf_1$ and $\Pbf_2$.}
        \end{subfigure}
        \caption{}
\end{figure}

We start by describing a vanilla version of our algorithm for maximizing $\F$ over the unit hypercube, termed as \emph{continuous randomized bi-greedy} (\Cref{alg:bigreedy}). This version assumes blackbox oracle access to algorithms for a few computations  involving univariate functions of the form $\F(.,\xv_{-i})$ (e.g. maximization over $[0,1]$, computing \cen{.}, etc.). We first prove that the vanilla algorithm finds a solution with an objective value of at least $\frac{1}{2}$ of the optimum. In \Cref{sec:poly-implement}, we show how to approximately implement these oracles in polynomial time when $\F$ is coordinate-wise Lipschitz.

\begin{algorithm}[ht]
\small
\caption{(Vanilla) Continuous Randomized Bi-Greedy}
\label{alg:bigreedy}
 \textbf{input}: function $\F:[0,1]^n\rightarrow [0,1]$~\;
 \textbf{output}: vector $\mathbf{\za}=(\za_1,\ldots,\za_n)\in[0,1]^n$~\;

Initialize $\Xbf\leftarrow (0,\ldots,0)$ and $\Ybf\leftarrow (1,\ldots,1)$~\;
\For{$i=1$ to $n$}{
 Find $\zu,\zl\in[0,1]$ such that $\begin{cases}
     \zl\in\underset{z\in[0,1]}{\argmax}~\F(z,\Ybf_{-i})\\
       \zu\in\underset{z\in[0,1]}{\argmax}~\F(z,\Xbf_{-i})
     \end{cases}$~\;
     
\If {$\zu\leq \zl$}
{$\za_i\leftarrow \zl~;$} 
\Else
{$\forall z\in[\zl,\zu]$, let $\begin{cases}
g(z)\triangleq\F(z,\Xbf_{-i})-\F(\zl,\Xbf_{-i}),\\
h(z)\triangleq\F(z,\Ybf_{-i})-\F(\zu,\Ybf_{-i}),
\end{cases}$~\;

Let $\alpha\triangleq g(\zu)$ and $\beta\triangleq h(\zl)$~;~~~~~~~~\tcp{note that $\mathtt{\alpha},\mathtt{\beta}\geq 0$}

 Let $\rv(z)\triangleq \left(g(z),h(z)\right)$ be a continuous two-dimensional curve in $[-1,\alpha]\times[-1,\beta]$~\;
 
Compute $\cen{\rv}$ (i.e. positive-orthant concave envelope of $\rv(t)$ as in \Cref{def:concave-envelope})~\;

Find point $\Pbf\triangleq~$intersection of $\cen{\rv}$ and the line $h'-\beta=g'-\alpha$ on g-h plane~\;

 Suppose $\Pbf=\lambda \Pbf_1+(1-\lambda)\Pbf_2$, where $\lambda\in[0,1]$ and $\Pbf_j=\rv(z^{(j)}), z^{(j)}\in[\zl,\zu]$ for $~~~~~~~~~~~~~j=1,2$, and both points are also on the \cen{$\rv$}~;~~~~~~\tcp{see \Cref{fig:pos-region}}
 
Randomly pick $\za_i$ such that $\begin{cases}
\za_i\leftarrow z^{(1)}~~~\textrm{with probablity}~\lambda\\
\za_i\leftarrow z^{(2)}~~~\textrm{o.w.}
\end{cases}$~\;}

 Let $X_i\leftarrow \za_i $ and $Y_i\leftarrow \za_i$~;~~~~~~\tcp{after this, $\Xbf$ and $\Ybf$ will agree on coordinate $i$}
}
\Return {$\mathbf{\za}=(\za_1,\ldots,\za_n)$}
\end{algorithm}

\begin{theorem}
\label{thm:vanilla-alg-weak}
If $\F(\cdot)$ is non-negative and continuous submodular (or equivalently is \texttt{\emph{weak DR-SM}}), then \Cref{alg:bigreedy} is a randomized $\frac{1}{2}$-approximation algorithm, i.e. returns $\mathbf{\za}\in[0,1]^n$ s.t. 
\[
2 \Ex{\F(\mathbf{\za})}\geq\F(\xob),~~~~~~~~~~\textrm{where $\xob\in\underset{\xbf\in[0,1]^n}{\argmax}~\F(\xbf)$ is the optimal solution.}
\]
\end{theorem}
\subsection{Analysis of the Continuous Randomized Bi-Greedy (proof of \Cref{thm:vanilla-alg-weak})}
\label{sec:analysis-bigreedy}

We start by defining these vectors, used in our analysis in the same spirit as \cite{buchbinder2015tight}: 
\begin{align*}
&i\in[n]:~\Xbf^{(i)}\triangleq(\za_1,\ldots,\za_i,0,0,\ldots,0),~~~~~~\Xbf^{(0)}\triangleq(0,\ldots,0)\\
&i\in[n]:~\Ybf^{(i)}\triangleq(\za_1,\ldots,\za_i,1,1,\ldots,1),~~~~~~\Ybf^{(0)}\triangleq(1,\ldots,1)\\
&i\in[n]:~\Obf^{(i)}\triangleq(\za_1,\ldots,\za_i,\xo_{i+1},\ldots,\xo_{n}),~~~\Obf^{(0)}\triangleq(\xo_1,\ldots,\xo_n)
\end{align*}
Note that $\Xbf^{(i)}$ and $\Ybf^{(i)}$ (or $\Xbf^{(i-1)}$ and $\Ybf^{(i-1)}$) are the values of $\Xbf$ and $\Ybf$ at the end of (or at the beginning of) the $i^{\textrm{th}}$ iteration of \Cref{alg:bigreedy}. In the remainder of this section, we give the high-level proof ideas and present some proof sketches. See the supplementary materials for the formal proofs.

\subsubsection{{Reduction to coordinate-wise zero-sum games.}} For each coordinate $i\in[n]$, we consider a sub-problem. In particular, define a two-player \emph{zero-sum game} played between the \emph{algorithm player} (denoted by \texttt{ALG}) and the \emph{adversary player} (denoted by \texttt{ADV}). \texttt{ALG} selects a (randomized) strategy $\za_i\in[0,1]$, and \texttt{ADV} selects a (randomized) strategy $\xo_i\in[0,1]$. Recall the descriptions of $g(z)$ and $h(z)$ at iteration $i$ of \Cref{alg:bigreedy},:
\[
g(z)=\F(z,\Xbf_{-i}^{(i-1)})-\F(\zl,\Xbf_{-i}^{(i-1)})~~,~~h(z)=\F(z,\Ybf_{-i}^{(i-1)})-\F(\zu,\Ybf_{-i}^{(i-1)}).
\]
We now define the utility of \texttt{ALG} (negative of the utility of  \texttt{ADV}) in our zero-sum game as follows:
\begin{equation}
\label{eq:payoff}
\val(\za_i,\xo_i)\triangleq \frac{1}{2}g(\za_i)+\frac{1}{2}h(\za_i)-\max\left(g(\xo_i)-g(\za_i),h(\xo_i)-h(\za_i)\right).
\end{equation} 
Suppose the expected utility of \texttt{ALG} is non-negative at the equilibrium of this game. In particular, suppose \texttt{ALG}'s randomized strategy $\za_i$ (in \Cref{alg:bigreedy}) guarantees that for every strategy $\xo_i$ of \texttt{ADV} the expected utility of \texttt{ALG} is non-negative. If this statement holds for all of the zero-sum games corresponding to different iterations $i\in[n]$, then \Cref{alg:bigreedy} is a $\frac{1}{2}$-approximation of the optimum.

\begin{lemma}
\label{lem:reduction-to-zero-sum}
If $\forall i\in[n]: \Ex{\val(\za_i,\xo_i)}\geq -\delta/n$ for constant $\delta>0$, then $2\Ex{\F(\mathbf{\za})}\geq \F(\xob)-\delta$. 
\end{lemma}
\begin{proof}[Proof sketch.]
Our bi-greedy approach, \'a la \cite{buchbinder2015tight}, revolves around analyzing the evolving values of three points: $\Xbf^{(i)}$, $\Ybf^{(i)}$, and $\Obf^{(i)}$. These three points begin at all-zeroes, all-ones, and the optimum solution, respectively, and converge to the algorithm's final point. In each iteration, we aim to relate the total increase in value of the first two points with the decrease in value of the third point. If we can show that the former quantity is at least twice the latter quantity, then a telescoping sum proves that the algorithm's final choice of point scores at least half that of optimum.
  
  The utility of our game is specifically engineered to compare the total increase in value of the first two points with the decrease in value of the third point. The positive term of the utility is half of this increase in value, and the negative term is a bound on how large in magnitude the decrease in value may be. As a result, an overall nonnegative utility implies that the increase beats the decrease by a factor of two, exactly the requirement for our bi-greedy approach to work. Finally, an additive slack of $\delta / n$ in the utility of each game sums over $n$ iterations for a total  slack of $\delta$.
\end{proof}

\begin{proof}[\textbf{Proof of \Cref{lem:reduction-to-zero-sum}.}]
Consider a realization of $\za_i$ where $\za_i\geq \xo_i$. We have:
\begin{align}
\F(\Obf^{(i-1)})-\F(\Obf^{(i)})&=\F(\za_1,\ldots,\za_{i-1},\xo_i,\xo_{i+1},\ldots,\xo_{n})-\F(\za_1,\ldots,\za_{i-1},\za_i,\xo_{i+1},\ldots,\xo_{n})\nonumber\\
&\leq \F(\za_1,\ldots,\za_{i-1},\xo_i,1,\ldots,1)-\F(\za_1,\ldots,\za_{i-1},\za_i,1,\ldots,1)\nonumber\\
&=\left(\F(\xo_i,\Ybf_{-i}^{(i-1)})-\F(\zu,\Ybf_{-i}^{(i-1)})\right)-\left(\F(\za_i,\Ybf_{-i}^{(i-1)})-\F(\zu,\Ybf_{-i}^{(i-1)})\right)\nonumber\\
&=h(\xo_i)-h(\za_i),\label{eq:opt-upper-1}
\end{align}
where the inequality holds due to \texttt{weak DR-SM}. Similarly, for a a realization of $\za_i$ where $\za_i\leq \xo_i$: 
\begin{align}
\F(\Obf^{(i-1)})-\F(\Obf^{(i)})&=\F(\za_1,\ldots,\za_{i-1},\xo_i,\xo_{i+1},\ldots,\xo_{n})-\F(\za_1,\ldots,\za_{i-1},\za_i,\xo_{i+1},\ldots,\xo_{n})\nonumber\\
&\leq \F(\za_1,\ldots,\za_{i-1},\xo_i,0,\ldots,0)-\F(\za_1,\ldots,\za_{i-1},\za_i,0,\ldots,0)\nonumber\\
&=\left(\F(\xo_i,\Xbf_{-i}^{(i-1)})-\F(\zl,\Xbf_{-i}^{(i-1)})\right)-\left(\F(\za_i,\Xbf_{-i}^{(i-1)})-\F(\zl,\Xbf_{-i}^{(i-1)})\right)\nonumber\\
&=g(\xo_i)-g(\za_i)\label{eq:opt-upper-2}
\end{align}
Putting \cref{eq:opt-upper-1} and \cref{eq:opt-upper-2} together, for every realization $\za_i$ we have: 
\begin{equation}
F(\Obf^{(i-1)})-\F(\Obf^{(i)})\leq \max\left(g(\xo_i)-g(\za_i),h(\xo_i)-h(\za_i)\right)\label{eq:telescope1}
\end{equation}
Moreover, consider the term $\F(\Xbf^{(i)})-\F(\Xbf^{(i-1)})$. We have: 
\begin{align}
\F(\Xbf^{i})-\F(\Xbf^{(i-1)})&=\F(\za_1,\ldots,\za_{i-1},\za_i,0,\ldots,0)-\F(\za_1,\ldots,\za_{i-1},0,0,\ldots,0)\nonumber\\
& =g(\za_i)-g(0)=g(\za_i)+\F(\zl,\Xbf_{-i}^{(i-1)})-\F(\Xbf^{(i-1)})\nonumber\\
&\geq g(\za_i)+\F(\zl,\Ybf_{-i}^{(i-1)})-\F(0,\Ybf_{-i}^{(i-1)})\geq g(\za_i)\label{eq:telescope2}
\end{align}
where the first inequality holds due to \texttt{weak DR-SM} property and the second inequity holds as $\zl\in\underset{z\in[0,1]}{\argmax}~\F(z,\Ybf_{-i}^{(i-1)})$. Similarly, consider the term $\F(\Ybf^{(i)})-\F(\Ybf^{(i-1)})$. We have:
\begin{align}
\F(\Ybf^{(i)})-\F(\Ybf^{(i-1)})&=\F(\za_1,\ldots,\za_{i-1},\za_i,1,\ldots,1)-\F(\za_1,\ldots,\za_{i-1},1,1,\ldots,1)\nonumber\\
& =h(\za_i)-h(1)=h(\za_i)+\F(\zu,\Ybf_{-i}^{(i-1)})-\F(\Ybf^{(i-1)})\nonumber\\
&\geq h(\za_i)+\F(\zu,\Xbf_{-i}^{(i-1)})-\F(1,\Xbf_{-i}^{(i-1)})\geq h(\za_i)\label{eq:telescope3}
\end{align}
where the first inequality holds due to \texttt{weak DR-SM} and the second inequity holds as $\zu\in\underset{z\in[0,1]}{\argmax}~\F(z,\Xbf_{-i}^{(i-1)})$. By \cref{eq:telescope1}, \cref{eq:telescope2}, \cref{eq:telescope3}, and the fact that $\F(\mathbf{0})+\F(\mathbf{1})\geq 0$, we have:
\begin{align*}
0&\leq \displaystyle\sum_{i=1}^{n}{\Ex{\val(\za_i,\xo_i)}}= \displaystyle\sum_{i=1}^{n}\left({\frac{1}{2}\Ex{g(\za_i)}+\frac{1}{2}\Ex{h(\za_i)}-\Ex{\max\left(g(\xo_i)-g(\za_i),h(\xo_i)-h(\za_i)\right)}}\right)\\
&\leq \frac{1}{2}\displaystyle\sum_{i=1}^{n}\left( \F(\Xbf^{(i)})-\F(\Xbf^{(i-1)})\right)+ \frac{1}{2}\displaystyle\sum_{i=1}^{n} \left(\F(\Ybf^{(i)})-\F(\Ybf^{(i-1)})\right)- \displaystyle\sum_{i=1}^{n}\left( \F(\Obf^{(i-1)})-\F(\Obf^{(i)})\right)\\
&=\frac{\F(\Xbf^{(n)})-\F(\Xbf^{(0)})}{2}+\frac{\F(\Ybf^{(n)})-\F(\Ybf^{(0)})}{2}-\F(\Obf^{(0)})+\F(\Obf^{(n)})\\
&\leq \frac{\F(\mathbf{\za})}{2}+ \frac{\F(\mathbf{\za})}{2}-\F(\xob)+\F(\mathbf{\za})=2\F(\mathbf{\za})-\F(\xob) \qedhere
\end{align*}\phantom\qedhere
\end{proof}

\subsubsection{{Analyzing the zero-sum games.}} Fix an iteration $i\in [n]$ of \Cref{alg:bigreedy}. We then have the following.

\begin{proposition}
\label{prop:game-analysis}
 If \texttt{ALG} plays the (randomized) strategy $\za_i$ as described in \Cref{alg:bigreedy}, then we have $ \Ex{\val(\za_i,\xo_i)}\geq 0$  against any strategy $\xo_i$ of \texttt{ADV}.
\end{proposition}
\begin{proof}[Proof of \Cref{prop:game-analysis}] We do the proof by case analysis over two cases:

\paragraph{$\square~$ Case $\mathbf{\zl\geq \zu}$ (\textit{easy}):} In this case, the algorithm plays a deterministic strategy  $\za_i=\zl$. We therefore have:
\begin{equation*}
\val(\za_i,\xo_i)= \frac{1}{2}g(\za_i)+\frac{1}{2}h(\za_i)-\max\left(g(\xo_i)-g(\za_i),h(\xo_i)-h(\za_i)\right)\geq \min (g(\za_i)-g(\xo_i),0)
\end{equation*}
where the inequality holds because $g(\za_i)=g(\zl)=0$, and also $\zl\in \underset{z\in[0,1]}{\argmax}~\F(z,\Ybf_{-i}^{(i)})$ and so:
\begin{itemize}
\item  $h(\za_i)=h(\zl)=\F(\zl,\Ybf_{-i}^{(i)})-\F(\zu,\Ybf_{-i}^{(i)})\geq 0$
\item $h(\xo_i)-h(\za_i)=\F(\xo_i,\Ybf_{-i}^{(i-1)})-\F(\zl,\Ybf_{-i}^{(i-1)})\leq 0$
\end{itemize}
To complete the proof for this case, it is only remained to show $g(\za_i)-g(\xo_i)\geq 0$. As $\zl\geq \zu$, for any given $\xo_i\in[0,1]$ either $\xo_i\leq \zl$ or $\xo_i\geq \zu$ (or both). If $\xo_i\leq \zl$ then:
\begin{align*}
g(\za_i)-g(\xo_i)=-g(\xo_i)=\F(\zl,\Xbf^{(i-1)}_{-i})-\F(\xo_i,\Xbf^{(i-1)}_{-i})\geq \F(\zl,\Ybf^{(i-1)}_{-i})-\F(\xo_i,\Ybf^{(i-1)}_{-i})\geq 0
\end{align*}
where the first inequality uses \texttt{weak DR-SM} property and the second inequality uses the fact $\zl\in \underset{z\in[0,1]}{\argmax}~\F(z,\Ybf_{-i}^{(i)})$. If $\xo_i\leq \zu$, we then have:
\begin{align*}
g(\za_i)-g(\xo_i)&=\F(\zl,\Xbf^{(i-1)}_{-i})-\F(\xo_i,\Xbf^{(i-1)}_{-i})\\
&= \F(\zl,\Xbf^{(i-1)}_{-i})-\F(\zu,\Xbf^{(i-1)}_{-i})+\F(\zu,\Xbf^{(i-1)}_{-i})-\F(\xo_i,\Xbf^{(i-1)}_{-i})\\
&\geq \left(\F(\zl,\Ybf^{(i-1)}_{-i})-\F(\zu,\Ybf^{(i-1)}_{-i})\right)+\left(\F(\zu,\Xbf^{(i-1)}_{-i})-\F(\xo_i,\Xbf^{(i-1)}_{-i})\right)\geq0
\end{align*}
where the first inequality uses \texttt{weak DR-SM} property and the second inequality holds because both terms are non-negative, following the fact that:
\[
\zl\in \underset{z\in[0,1]}{\argmax}~\F(z,\Ybf_{-i}^{(i)})~~~~~~~~ \textrm{and}~~~~~~~~\zu\in \underset{z\in[0,1]}{\argmax}~\F(z,\Xbf_{-i}^{(i)})\]
 Therefore, we finish the proof of the easy case.

\paragraph{$\square~$ Case $\mathbf{\zl< \zu}$~(\textit{hard}):} In this case, \texttt{ALG} plays a mixed strategy over two points. To determine the two-point support, it considers the curve $\rv=\{(g(z),h(z))\}_{z\in[\zl,\zu]}$ and finds a point $\Pbf$ on $\cen{\rv}$ (\ie\Cref{def:concave-envelope}) that lies on the line $h'-\beta=g'-\alpha$, where recall that $\alpha=g(\zu)\geq 0$ and $\beta=g(\zl)\geq 0$ (as $\zu$ and $\zl$ are the maximizers of $\F(z,\Xbf_{-i}^{(i-1)})$ and  $\F(z,\Ybf_{-i}^{(i-1)})$ respectively). Because this point is on the concave envelope it should be a convex combination of two points on the curve $\rv(z)$.  Lets say $\Pbf=\lambda\Pbf_1+(1-\lambda)\Pbf_2$, where $\Pbf_1=\rv(z^{(1)})$ and $\Pbf_2=\rv(z^{(2)})$, and $\lambda\in[0,1]$. The final strategy of \texttt{ALG} is a mixed strategy over $\{z^{(1)},z^{(2)}\}$ with probabilities $(\lambda,1-\lambda)$. Fixing any mixed strategy of \texttt{ALG} over two points $\Pbf_1=(g_1,h_1)$ and $\Pbf_2=(g_2,h_2)$ with probabilities $(\lambda,1-\lambda)$ (denoted by $\mathcal{F_{\Pbf}}$), define the \texttt{ADV}'s \emph{positive region}, i.e.
\[
(g',h')\in[-1,1]\times[-1,1]:~~ \mathbf{E}_{(g,h)\sim F_{\Pbf}}\left[{\frac{1}{2}g+\frac{1}{2}h-\max(g'-g,h'-h)}\right]\geq 0. 
\]
Now, suppose \texttt{ALG} plays a mixed strategy with the property that its corresponding \texttt{ADV}'s positive region covers the entire curve $\{g(z),h(z)\}_{z\in[0,1]}$. Then, for any strategy $\xo_i$ of  \texttt{ADV} the expected utility of \texttt{ALG} is non-negative. In the rest of the proof, we geometrically characterize the
\texttt{ADV}'s positive region against a mixed strategy of \texttt{ALG} over a 2-point support, and then we show for the particular choice of $\Pbf_1$, $\Pbf_2$ and $\lambda$ in \Cref{alg:bigreedy} the positive region covers the entire curve $\{g(z),h(z)\}_{z\in[0,1]}$.


\begin{lemma} 
\label{lem:pos-region}
Suppose \texttt{ALG} plays a 2-point mixed strategy over $\Pbf_1=\rv(z^{(1)})=(g_1,h_1)$ and $\Pbf_2=\rv(z^{(1)})=(g_2,h_2)$ with probabilities $(\lambda,1-\lambda)$, and w.l.o.g. $h_1-g_1\geq h_2-g_2$. Then \texttt{ADV}'s positive region is the pentagon  $(\mathcal{M}_0,\mathcal{M}_1,\mathcal{Q}_1,\mathcal{Q}_2,\mathcal{M}_2)$, where $\mathcal{M}_0=(-1,-1)$ and (see \Cref{fig:pos-region}):
\begin{enumerate}
\item $\mathcal{M}_1=\left(-1,\lambda(\frac{3}{2}h_1+\frac{1}{2}g_1)+(1-\lambda)(\frac{3}{2}h_2+\frac{1}{2}g_2)\right)$,
\item $\mathcal{M}_2=\left(\lambda(\frac{3}{2}g_1+\frac{1}{2}h_1)+(1-\lambda)(\frac{3}{2}g_2+\frac{1}{2}h_2),-1\right)$,
\item $\mathcal{Q}_1$ is the intersection of the lines leaving $\Pbf_1$ with slope $1$ and leaving $\mathcal{M}_1$ along the g-axis,  
\item $\mathcal{Q}_2$ is the intersection of the lines leaving $\Pbf_2$ with slope $1$  and leaving $\mathcal{M}_2$ along the h-axis.
\end{enumerate}
\end{lemma}
\begin{proof}[\textbf{Proof of \Cref{lem:pos-region}.}]
 We start by a technical lemma, showing a single-crossing property of the g-h curve of a \texttt{weak DR} submodular function $\F(.)$, and we then characterize the region using this lemma.
\begin{lemma}\label{lem:single-crossing}
The univariate function $d(z)=h(z)-g(z)$ is monotone non-increasing. 
\end{lemma}
\begin{proof}
By using \texttt{weak DR-SM} property of $\F(.)$ the proof is immediate, as for any $\delta\geq 0$,
$$d(z+\delta)-d(z)=\left(\F(z+\delta,\Ybf^{(i-1)}_{-i})-\F(z,\Ybf^{(i-1)}_{-i})\right)-(F(z+\delta,\Xbf^{(i-1)}_{-i})-F(z+\delta,\Xbf^{(i-1)}_{-i}))\leq 0,$$
where the inequality holds due to the fact that $\Ybf^{(i-1)}_{-i}\geq \Xbf^{(i-1)}_{-i}$ and $\delta\geq 0$.
\end{proof}
Being equipped with \Cref{lem:single-crossing}, the positive region is the set of all points $(g',h')\in[-1,1]^2$ such that 
\begin{align*}
&\mathbf{E}_{(g,h)\sim F_{\Pbf}}\left[{\frac{1}{2}g+\frac{1}{2}h-\max(g'-g,h'-h)}\right]\\
&=\lambda \left(\frac{1}{2}g_1+\frac{1}{2}h_1-\max(g'-g_1,h'-h_1)\right)+(1-\lambda)\left( \frac{1}{2}g_2+\frac{1}{2}h_2-\max(g'-g_2,h'-h_2)\right)\geq 0
\end{align*}
The above inequality defines a polytope. Our goal is to find the vertices and faces of this polytope. Now, to this end, we only need to consider three cases: 1) $h'-g'\geq h_1-g_1$, 2) $h_2-g_2\leq h'-g'\leq h_1-g_1$ and 3) $h'-g'\leq h_2-g_2$ (note that $h_1-g_1\geq h_2-g_2$). From the first and third case we get the half-spaces $h'\leq  \lambda(\frac{3}{2}h_1+\frac{1}{2}g_1)+(1-\lambda)(\frac{3}{2}h_2+\frac{1}{2}g_2)$ and $g'\leq  \lambda(\frac{3}{2}g_1+\frac{1}{2}h_1)+(1-\lambda)(\frac{3}{2}g_2+\frac{1}{2}h_2)$ respectively, that form two of the faces of the positive-region polytope. From the second case, we get another half-space, but the observation is that the transition from first case to second case happens when $h'-g'=h_1-g_1$, i.e. on a line with slope one leaving $\Pbf_1$, and transition from second case to the third case happens when $h'-g'=h_2-g_2$, i.e. on a line with slope one leaving $\Pbf_2$. Therefore, the second half-space is the region under the line connecting two points $\mathcal{Q}_1$ and $\mathcal{Q}_2$, where $\mathcal{Q}_1$ is the intersection of $h'=\lambda(\frac{3}{2}h_1+\frac{1}{2}g_1)+(1-\lambda)(\frac{3}{2}h_2+\frac{1}{2}g_2)$ and the line leaving $\Pbf_1$ with slope one (point $\mathcal{Q}_1$), and $\mathcal{Q}_2$ is the intersection of $g'=\lambda(\frac{3}{2}g_1+\frac{1}{2}h_1)+(1-\lambda)(\frac{3}{2}g_2+\frac{1}{2}h_2)$ and the line leaving $\Pbf_2$ with slope one (point $\mathcal{Q}_2$). The line segment $\mathcal{Q}_1-\mathcal{Q}_2$ defines another face of the positive region polytope, and $\mathcal{Q}_1$ and $\mathcal{Q}_2$ will be two vertices on this face. By intersecting the three mentioned half-spaces with $g'\geq -1$ and $h\geq -1$ (which define the two remaining faces of the positive region polytope), the postive region will be the polytope defined by the pentagon $(\mathcal{M}_0,\mathcal{M}_1,\mathcal{Q}_1,\mathcal{Q}_2,\mathcal{M}_2)$, as claimed (see \Cref{fig:pos-region} for a pictorial proof).
\end{proof} 

By applying \Cref{lem:pos-region}, we have the following main technical lemma. The proof is geometric and is pictorially visible in \Cref{fig:pos-region}.  This lemma finishes the proof of \Cref{prop:game-analysis}.

\begin{lemma}[\textbf{\emph{main lemma}}]
\label{lem:main-technical}
If \texttt{ALG} plays the two point mixed strategy described in \Cref{alg:bigreedy}, then for every $\xo_i\in[0,1]$ the point $(g',h')=\left(g(\xo_i),h(\xo_i)\right)$ is in the \texttt{ADV}'s positive region. 
\end{lemma}

\begin{proof}[Proof sketch.]
For simplicity assume $\zl=0$ and $\zu=1$. To understand the \texttt{ADV}'s positive region that results from playing a two-point mixed strategy by \texttt{ALG}, we consider the positive region that results from playing a one point pure strategy. When \texttt{ALG} chooses a point $(g,h)$, the positive term of the utility is one-half of its one-norm. The negative term of the utility is the worse between how much the \texttt{ADV}'s point is above \texttt{ALG}'s point, and how much it is to the right of \texttt{ALG}'s point. The resulting positive region is defined by an upper boundary $g' \le \frac32 g + \frac12 h$ and a right boundary $h' \le \frac12 g + \frac32 h$.
  
Next, let's consider what happens when we pick point $(g_1, h_1)$ with probability $\lambda$ and point $(g_2, h_2)$ with probability $(1 - \lambda)$. We can compute the expected point: let $(g_3, h_3) = \lambda(g_1, h_1) + (1 - \lambda)(g_2, h_2)$. As suggested by Lemma~\ref{lem:pos-region}, the positive region for our mixed strategy has three boundary conditions: an upper boundary, a right boundary, and a corner-cutting boundary. The first two boundary conditions correspond to a pure strategy which picks $(g_3, h_3)$. By design, $(g_3, h_3)$ is located so that these boundaries cover the entire $[-1, \alpha] \times [-1, \beta]$ rectangle. This leaves us with analyzing the corner-cutting boundary, which is the focus of Figure~\ref{fig:pos-region}. As it turns out, the intersections of this boundary with the two other boundaries lie on lines of slope $1$ extending from $(g_j, h_j)_{ j=1,2}$. If we consider the region between these two lines, the portion under the envelope (where the curve $\rv$ may lie) is distinct from the portion outside the corner-cutting boundary. However, if $\rv$ were to ever violate the corner-cutting boundary condition without violating the other two boundary conditions, it must do so in this region. Hence the resulting positive region covers the entire curve $\rv$, as desired.
\end{proof}
\begin{proof}[\textbf{Proof of \Cref{lem:main-technical}}]
First of a all, we claim any \texttt{ADV}'s strategy $\xo_i\in[0,\zl)$ (or $\xo_i\in(\zu,1]$) is weakly dominated by $\zl$ (or $\zu$) if \texttt{ALG} plays a (randomized) strategy $\za_i\in[\zl,\zu]$. To see this, if $\xo_i\in[0,\zl)$, 
\begin{align*}
&\max\left(g(\xo_i)-g(\za_i),h(\xo_i)-h(\za_i)\right)\\
&=\max\left(\F(\xo_i,\Xbf^{(i-1)}_{-i})-\F(\za_i,\Xbf^{(i-1)}_{-i}),\F(\xo_i,\Ybf^{(i-1)}_{-i})-\F(\za_i,\Ybf^{(i-1)}_{-i})\right)\\
&=\F(\xo_i,\Ybf^{(i-1)}_{-i})-\F(\za_i,\Ybf^{(i-1)}_{-i})\leq \F(\zl,\Ybf^{(i-1)}_{-i})-\F(\za_i,\Ybf^{(i-1)}_{-i})\\
&=h(\zl)-h(\za_i)\leq \max\left(g(\zl)-g(\za_i),h(\zl)-h(\za_i)\right)
\end{align*}
and therefore $\val(\za_i,\zl)\leq \val(\za_i,\xo_i)$ for any $\xo_i\in[0,\zl)$. Similarly, $\val(\za_i,\zu)\leq \val(\za_i,\xo_i)$ for any $\xo_i\in(\zu,1]$. So, without loss of generality, we can assume \texttt{ADV}'s strategy $\xo_i$ is in $[\zl,\zu]$. 

Now, consider the curve $\rv=\{(g(z),h(z)\}_{z\in[\zl,\zu]}$ as in \Cref{fig:pos-region}. \texttt{ALG}'s strategy is a 2-point mixed strategy over $\Pbf_1=(g_1,h_1)=\rv(z^{(1)})$ and $\Pbf_2=(g_2,h_2)=\rv(z^{(1)})$, where these two points are on different sides of the line $\mathcal{L}: h'-\beta=g'-\alpha$ (or both of them are on the line $\mathcal{L}$). Without loss of generality, assume $h_1-g_1\geq \beta-\alpha\geq h_2-g_2$. Note that $\rv(\zl)=(0,\beta)$ is above the line $\mathcal{L}$ and $\rv(\zl)=(\alpha,0)$ is below the line $\mathcal{L}$. So, because $h(z)-g(z)$ is monotone non-increasing due to \Cref{lem:single-crossing}, we should have $\zl\leq z^{(1)}\leq z^{(2)}\leq \zu$.

Using \Cref{lem:pos-region}, the \texttt{ADV}'s positive region is  $(\mathcal{M}_0,\mathcal{M}_1,\mathcal{Q}_1,\mathcal{Q}_2,\mathcal{M}_2)$, where $\{\mathcal{M}_j\}_{j=1,2,3}$ and $\{\mathcal{Q}_j\}_{j=1,2}$ are as described in \Cref{lem:pos-region}. The upper concave envelope $\cen{\rv}$ upper-bounds the curve $\rv$. Therefore, to show that curve $\rv$ is entirely covered by the \texttt{ADV}'s positive region, it is only enough to show its upper concave envelope $\cen{\rv}$ is entirely covered (as can also be seen from \Cref{fig:pos-region}). Lets denote the line leaving $\Pbf_j$ with slope one by $\mathcal{L}_j$ for $j=1,2$. The curve $\cen{\rv}$ consists of three parts: the part above $\mathcal{L}_1$, the part below $\mathcal{L}_2$ and the part between $\mathcal{L}_1$ and $\mathcal{L}_2$ (the last part is indeed the line segment connecting $\Pbf_1$ and $\Pbf_2$). Interestingly, the line connecting $\Pbf_1$ to $\mathcal{Q}_1$ and the line connecting $\Pbf_2$ to $\mathcal{Q}_2$ both have slope $1$. So, as it can be seen from \Cref{fig:pos-region}, if we show $\mathcal{Q}_1$ is above the line $h'=\beta$ and $\mathcal{Q}_2$ is to the right of the line $g'=\alpha$, then the $\cen{\rv}$ will entirely be covered by the positive region and we are done. To see why this holds, first note that $\lambda$ has been picked so that $\Pbf \triangleq(\Pbf_g,\Pbf_h)=\lambda\Pbf_1+(1-\lambda)\Pbf_2)$. Due to \Cref{lem:pos-region}, 
\begin{align*}
\mathcal{Q}_1,h&= \lambda(\frac{3}{2}h_1+\frac{1}{2}g_1)+(1-\lambda)(\frac{3}{2}h_2+\frac{1}{2}g_2)=\frac{3}{2}\Pbf_h+\frac{1}{2}\Pbf_g\\
\mathcal{Q}_2,g&= \lambda(\frac{3}{2}g_1+\frac{1}{2}h_1)+(1-\lambda)(\frac{3}{2}g_2+\frac{1}{2}h_2)=\frac{3}{2}\Pbf_g+\frac{1}{2}\Pbf_h
\end{align*}
Moreover, point $\Pbf=(\Pbf_g,\Pbf_h)$ dominates the point $\mathcal{C}\triangleq (\frac{\alpha^2}{\alpha+\beta},\frac{\beta^2}{\alpha+\beta}) $ coordinate-wise. This dominance is simply true because points $\mathcal{C}$ and $\Pbf$ are actually the intersections of the line $\mathcal{L}: h'-\beta=g'-\alpha$ (with slope one) with the line connecting $(0,\beta)$ to $(\alpha,0)$ and with the curve $\cen{\rv}$ respectively. As $\cen{\rv}$ upper-bounds the line connecting $(0,\beta)$ to $(\alpha,0)$, and because $\mathcal{L}$ has slope one, $\Pbf_h\geq \mathcal{C}_h=\frac{\beta^2}{\alpha+\beta}$ and $\Pbf_g\geq \mathcal{C}_g=\frac{\alpha^2}{\alpha+\beta}$. Putting all the pieces together,
\begin{align*}
\mathcal{Q}_{1,h}&\geq \frac{3}{2}\frac{\beta^2}{\alpha+\beta}+\frac{1}{2}\frac{\alpha^2}{\alpha+\beta}=\frac{\left(\alpha^2+\beta^2-2\alpha\beta\right)+2\beta^2+2\alpha\beta}{2(\alpha+\beta)}=\beta+\frac{(\alpha-\beta)^2}{2(\alpha+\beta)}\geq \beta\\
\mathcal{Q}_{2,g}&\geq \frac{3}{2}\frac{\alpha^2}{\alpha+\beta}+\frac{1}{2}\frac{\beta^2}{\alpha+\beta}=\frac{\left(\alpha^2+\beta^2-2\alpha\beta\right)+2\alpha^2+2\alpha\beta}{2(\alpha+\beta)}=\alpha+\frac{(\alpha-\beta)^2}{2(\alpha+\beta)}\geq \alpha
\end{align*}
which implies $\mathcal{Q}_1$ is above the line $h'=\beta$ and $\mathcal{Q}_2$ is to the right of the line $g'=\alpha$, as desired.
\end{proof}
\noqed
\end{proof}
\vspace{-3mm}
\subsection{Polynomial-time Implementation under Lipschitz Continuity: Overview}
\label{sec:poly-implement}

At each iteration, \Cref{alg:bigreedy} interfaces with $\F$ in two ways: (i) when performing optimization to compute $\zl, \zu$ and (ii) when computing the upper-concave envelope. In both cases, we are concerned with univariate projections of $\F$, namely $\F(z,\Xbf_{-i})$ and $\F(z,\Ybf_{-i}$. Assuming $\F$ is coordinate-wise Lipschitz continuous with constant $C>0$, we choose a small $\epsilon > 0$ and take periodic samples at $\epsilon$-spaced intervals from each one of these functions, for a total of $O(\frac{1}{\epsilon})$ samples. 

To perform task (i), we simply return the the sample which resulted in the maximum function value. Since the actual maximum is $\epsilon$-close to one of the samples, our maximum is at most an additive $\epsilon C$ lower in value. To perform task (ii), we use these samples to form an approximate $\rv(z)$ curve, denoted by $\hat{\rv}(z)$. Note that we then proceed exactly as described in \Cref{alg:bigreedy} to pick a (randomized) strategy $\za_i$ using $\hat{\rv}(z)$. Note that \texttt{ADV} can actually choose a point on the exact curve $\rv(z)$. However the point she chooses is close to one of our samples and hence is at most an additive $\epsilon C$ better in value with respect to functions $g(.)$ and $h(.)$. Furthermore, we can compute the upper-concave envelope $\hat{\rv}(z)$ in time linear in the number of samples using Graham's algorithm~\citep{graham1972efficient}. Roughly speaking, this is because we can go through the samples in order of $z$-coordinate, avoiding the sorting cost of running Graham's on completely unstructured data. Formally, we have the following proposition. For detailed implementations, see \Cref{alg:apx-opt} and \Cref{alg:apx-ce}.
\begin{proposition}
If $\F$ is coordinate-wise Lipschitz continuous with constant $C>0$, then there exists an implementation of \Cref{alg:bigreedy} that runs in time $O(n^2/\epsilon)$ and returns a (randomized) point $\mathbf{\za}$ s.t. 
\[
2 \Ex{\F(\mathbf{\za})}\geq\F(\xob)-2C\epsilon,~~~~~~~~~~\textrm{where $\xob\in\underset{\xbf\in[0,1]^n}{\argmax}~\F(\xbf)$ is the optimal solution.}
\]
\end{proposition}
\begin{algorithm}[ht]
\small
\caption{Approximate One-Dimensional Optimization}
\label{alg:apx-opt}
 \textbf{input}: function $f : [0,1] \rightarrow [0,1]$, additive error $\delta > 0$, Lipschitz Constant $C > 0$~\;
 \textbf{output}: coordinate value $z \in [0,1]^n$~\;
 
Set $\epsilon \leftarrow \frac{\delta}{C}$~\;
Initialize $z^* \leftarrow 0$~\;
Initialize $z \leftarrow 0$~\;
\While{$z \le 1$}{
  \If{$f(z) > f(z^*)$}{
    $z^* \leftarrow z$~\;
  }
  $z \leftarrow z + \epsilon$~\;
}
\Return{$z^*$}
\end{algorithm}

\begin{algorithm}[ht]
\small
\caption{Approximate Annotated Upper-Concave Envelope}
\label{alg:apx-ce}
 \textbf{input}: function $f : [0,1] \rightarrow [0,1]$, function $g : [0, 1] \rightarrow [0, 1]$, additive error $\delta > 0$, Lipschitz Constant $C > 0$~\;
 \textbf{output}: stacks $s$ and $t$~\;
 
Set $\epsilon \leftarrow \frac{\delta}{C}$~\;
Initialize stacks $s, t$~\;
Initialize $z \leftarrow 0$~\;
\While{$z \le 1$}{
  \If{$s$ is empty or $f(z)$ is strictly larger than the first coordinate of the the top element of $s$}{
    \While{$s$ has at least two elements and the slope from (the second-to-top element of $s$) to (the top element of $s$) is less than the slope from (the top element of $s$) to $(f(z), g(z))$}{
      Pop the top element of $s$~\;
      Pop the top element of $t$~\;
    }
    Push $(f(z), g(z))$ onto $s$~\;
    Push $z$ onto $t$~\;
  }
  $z \leftarrow z + \epsilon$~\;
}
\Return{$(s, t)$}
\end{algorithm}
\section{Strong DR-SM Maximization: Binary-Search Bi-Greedy}
\label{sec:strong}
Our second result is a fast binary search algorithm, achieving the tight $\frac{1}{2}$-approximation factor (up to additive error $\delta$) in quasi-linear time in $n$, but only for the special case of \texttt{strong DR-SM} functions (a.k.a. DR-submodular); see \Cref{def:weak-strong-DR}. This algorithm leverages the coordinate-wise concavity to identify a coordinate-wise \emph{monotone equilibrium condition}. In each iteration, it hunts for an equilibrium point by using binary search. Satisfying the equilibrium at each iteration then guarantees the desired approximation factor. Formally we propose \Cref{alg:binary-search}.
\begin{algorithm}[h]
\small
\caption{Binary-Search Continuous Bi-greedy}
\label{alg:binary-search}
  
   \textbf{input}: function $\F:[0,1]^n\rightarrow [0,1]$, error $\epsilon>0$~\;
 \textbf{output}: vector $\mathbf{\za}=(\za_1,\ldots,\za_n)\in[0,1]^n$~\;

Initialize $\Xbf\leftarrow (0,\ldots,0)$ and $\Ybf\leftarrow (1,\ldots,1)$~\;
  \For{$i = 1$ to $n$}{
  \If{$\dvx{\F}{x_i}(0,\Xbf_{-i})< 0~~\&~~\dvx{\F}{x_i}(1,\Ybf_{-i})\leq 0$}{$\za_i\leftarrow 0$}
  \ElseIf{$\dvx{\F}{x_i}(0,\Xbf_{-i})\geq 0~~\&~~\dvx{\F}{x_i}(1,\Ybf_{-i})>0$}{$\za_i\leftarrow 1$}
  \Else{ \tcp{we do binary search.}
    \While{$Y_i - X_i > {\epsilon}/{n}$}{
      Let $\za_i \leftarrow \frac{X_i + Y_i}{2}$~\;
  
      \eIf{$\dvx{\F}{x_i}(\za_i,\Xbf_{-i}) \cdot (1-\za_i) +\dvx{\F}{x_i}(\za_i,\Ybf_{-i})\cdot \za_i<0$}{
        \tcp{we need to increase $w_i$.}
        Set $X_i \leftarrow \za_i$~\;
      } {
        \tcp{we need to decrease $w_i$.}
        Set $Y_i \leftarrow \za_i$~\;
      }
    }}Let $X_i\leftarrow \za_i $ and $Y_i\leftarrow \za_i$~; ~~~~~~\tcp{after this, $\Xbf$ and $\Ybf$ will agree at coordinate $i$}
 
  }
  \Return{$\mathbf{\hat{z}}=(\za_1,\ldots,\za_n)$}
\end{algorithm}
As a technical assumption, let $\F$ be Lipschitz continuous with some constant $C>0$, so that we can relate the precision of our binary search with additive error. We arrive at the theorem, proved in \Cref{sec:strong-analysis}.

\begin{theorem}
\label{thm:alg-strong}
If $\F(.)$ is non-negative and DR-submodular (a.k.a \texttt{\emph{Strong DR-SM}}) and is coordinate-wise Lipschitz continuous with constant $C>0$, then \Cref{alg:binary-search} runs in time $O\left(n\log(\frac{n}{\epsilon})\right)$ and is a deterministic $\frac{1}{2}$-approximation algorithm up to $O(\epsilon)$ additive error, i.e. returns $\mathbf{\za}\in[0,1]^n$ s.t. 
\[
2{\F(\mathbf{\za})}\geq \F(\xob)-2C\epsilon~,~~~~~~~~~~\textrm{where $\xob\in\underset{\xbf\in[0,1]^n}{\argmax}~\F(\xbf)$ is the optimal solution.}
\]
\end{theorem}

\paragraph{Running time.} If we show that  $f(z)\triangleq \dvx{\F}{x_i}(z,\Xbf_{-i})(1-z) +\dvx{\F}{x_i}(z,\Ybf_{-i})z$ is monotone non-increasing in $z$, then clearly the binary search terminates in $O\left(\log(n/\epsilon)\right)$ steps (note that the algorithm only does binary search in the case when $f(0)\geq 0$ and $f(1)\leq 0$). To see the monotonicity,
$$ f'(z)=(1-z)\dvxx{\F}{x_i}(z,\Xbf_{-i})+z\dvxx{\F}{x_i}(z,\Ybf_{-i})+\left(\dvx{\F}{x_i}(z,\Ybf_{-i})-\dvx{\F}{x_i}(z,\Xbf_{-i})\right)\leq 0$$
where the inequality holds due to \texttt{strong DR-SM} and the fact that all of the Hessian entries (including diagonal) are non-positive. Hence the total running time is $O\left(n\log(n/\epsilon)\right)$.

\subsection{Analysis of the Binary-Search Bi-Greedy (proof of \Cref{thm:alg-strong})}
\label{sec:strong-analysis}
We start by the following technical lemma, which is used in various places of our analysis. The proof is immediate by \texttt{strong DR-SM} property (\Cref{def:weak-strong-DR}).
\begin{lemma}\label{lem:alpha-beta}
   For any $\mathbf{y},\mathbf{z}\in[0,1]^n$ such that $\mathbf{y}\leq\mathbf{z}$, we have $\dvx{\F}{x_i}(\mathbf{y})-\dvx{\F}{x_i}(\mathbf{z}) \geq 0, \forall i$.
\end{lemma}
\begin{proof}[\textbf{Proof of \Cref{lem:alpha-beta}.}] We rewrite this difference as a sum over integrals of the second derivatives:
  \begin{align*}
    \dvx{\F}{x_i}(\mathbf{y})-\dvx{\F}{x_i}(\mathbf{z})
      &= \sum_{j=1}^{n}
           \left[ \begin{aligned}
             & \dvx{\F}{x_i}(y_1, \ldots, y_{j-1}, y_j, z_{j+1}, \ldots, z_n) \\
           - & \dvx{\F}{x_i}(y_1, \ldots, y_{j-1}, z_j, z_{j+1}, \ldots, z_n)
           \end{aligned} \right]
         \\
      &= \sum_{j=1}^{n} \int_{y_j}^{z_j}
           - \dvxy{\F}{x_i}{x_j}(y_1, \ldots, y_{j-1}, w, z_{j+1}, \ldots, z_n) dw\geq 0
  \end{align*}
To see why the last inequality holds, because of the \texttt{strong DR-SM} \Cref{prop:weak-strong-multi-defs} implies that all of the second derivatives of $\F$ are always nonpositive. As $\forall i: z_i\geq y_i$, the RHS is nonnegative.
\end{proof}
\paragraph{\emph{A modified zero-sum game.}} We follow the same approach and notations as in the proof of \Cref{thm:vanilla-alg-weak} (\Cref{sec:analysis-bigreedy}). Suppose $\xob$ is the optimal solution. For each coordinate $i$ we again define a two-player zero-sum game between \texttt{ALG} and \texttt{ADV}, where the former plays $\za_i$ and the latter plays $\xo_i$. The payoff matrix for the \texttt{strong DR-SM} case, denoted by $\val_S(\za_i,\xo_i)$ is defined as before (\Cref{eq:payoff}); the only major difference is we redefine $h(.)$ and $g(.)$ to be the following functions,:
$$g(z)\triangleq\F(z,\Xbf_{-i}^{(i-1)})-\F(0,\Xbf_{-i}^{(i-1)})~~,~~h(z)\triangleq\F(z,\Ybf_{-i}^{(i-1)})-\F(1,\Ybf_{-i}^{(i-1)}).$$
Now, similar to \Cref{lem:reduction-to-zero-sum}, we have a lemma that shows how to prove the desired approximation factor using the above zero-sum game. The proof is exactly as  \Cref{lem:reduction-to-zero-sum} and is omitted for brevity.
\begin{lemma}
\label{lem:reduction-to-zero-sum-strong}
Suppose $\forall i\in[n]: \val_S(\za_i,\xo_i)\geq -\delta/n$ for constant $\delta>0$. Then $2{\F(\mathbf{\za})}\geq \F(\xob)-\delta$. 
\end{lemma}
\paragraph{\emph{Analyzing zero-sum games.}} We show that $\val_S(\za_i,\xo_i)$ is lower-bounded by a small constant, and then by using \Cref{lem:reduction-to-zero-sum-strong} we finish the proof. The formal proof, which appears in the supplementary materials, uses both ideas similar to those of \cite{buchbinder2015tight}, as well as new ideas on how to relate the algorithm's equilibrium condition to the value of the two-player zero-sum game. 
\begin{proposition}
\label{prop:main-strong}
if \texttt{ALG} plays the strategy $\za_i$ described in \Cref{alg:binary-search}, then $\val_S(\za_i,\xo_i)\geq -2C\epsilon/n$.
\end{proposition}

\begin{proof}[\textbf{Proof of \Cref{prop:main-strong}}]
Consider the easy case where $\dvx{\F}{x_i}(0,\Xbf^{(i-1)}_{-i})< 0$ and $\dvx{\F}{x_i}(1,\Ybf^{(i-1)}_{-i})\leq 0$. In this case, $\za_i=0$ and hence $g(\za_i)=g(0)=0$. Moreover, because of the \texttt{Strong DR-SM} property,
\begin{align*}
&h(0)=\F(0,\Ybf^{(i-1)}_{-i})-\F(1,\Ybf^{(i-1)}_{-i})\geq -\dvx{\F}{x_i}(1,\Ybf^{(i-1)}_{-i})\geq 0, \\
&h(\xo_i)-h(0)\leq g(\xo_i)-g(0)\leq \xo_i\cdot\dvx{\F}{x_i}(0,\Xbf^{(i-1)}_{-i})\leq 0, 
\end{align*}
and therefore  $\val_S(\za_i,\xo_i)=\frac{1}{2}g(0)+\frac{1}{2}h(0)-\max\left(g(\xo_i)-h(0),h(\xo_i)-h(0)\right)\geq 0$. The other easy case is when $\dvx{\F}{x_i}(0,\Xbf^{(i-1)}_{-i})\geq 0$ and $\dvx{\F}{x_i}(1,\Ybf^{(i-1)}_{-i})>0$. In this case $\za_i=1$ and a similar proof shows $\val_S(1,\xo_i)\geq 0$. 

Note that because of \Cref{lem:alpha-beta} we have $\dvx{\F}{x_i}(0,\Xbf^{(i-1)}_{-i})-\dvx{\F}{x_i}(1,\Ybf^{(i-1)}_{-i})\geq 0$, and hence the only remaining case (the not-so-easy one) is when $\dvx{\F}{x_i}(0,\Xbf^{(i-1)}_{-i})\geq 0$ and $\dvx{\F}{x_i}(1,\Ybf^{(i-1)}_{-i})\leq 0$. In this case,  \Cref{alg:binary-search} runs the binary search and ends up at a point $\za_i$. Because of the monotonicity and continuity of the equilibrium condition of the binary search, there exists $\tilde{z}$ that is $(\epsilon/n)$-close to $\za_i$ and $\dvx{\F}{x_i}(\tilde{z},\Xbf_{-i})(1- \tilde{z})+\dvx{\F}{x_i}(\tilde{z},\Ybf_{-i})\tilde{z}=0$. By a straightforward calculation using the Lipschitz continuity of $\F$ with constant $C$ and knowing that $\lvert\tilde{z}-\za_i \rvert\leq \epsilon/n$, we have: 
$$\val_S(\za_i,\xo_i)= \frac{1}{2}g(\za_i)+\frac{1}{2}h(\za_i)-\max\left(g(\xo_i)-g(\za_i),h(\xo_i)-h(\za_i)\right)\geq \val_S(\tilde{z},\xo_i)-\frac{2C\epsilon}{n}$$
So, we only need to show $\val_S(\tilde{z},\xo_i)\geq 0$. Let $\alpha\triangleq \dvx{\F}{x_i}(\tilde{z},\Xbf^{(i-1)}_{-i})$ and $\beta\triangleq-\dvx{\F}{x_i}(\tilde{z},\Ybf^{(i-1)}_{-i})$. Because of \Cref{lem:alpha-beta}, $\alpha+\beta\geq 0$. Moreover, $\alpha(1-\tilde{z})=\beta\tilde{z}$, and therefore we should have $\alpha\geq 0$ and $\beta\geq 0$. We now have two cases:
\paragraph{Case 1 ($\mathbf{\tilde{z}\geq \xo_i}$):} $g(\xo_i)-g(\tilde{z})\leq h(\xo_i)-h(\tilde{z})$ due to \texttt{strong DR-SM} and that $\tilde{z}\geq \xo_i$, so:
\begin{align*}
\val_S(\tilde{z},\xo_i)&=\frac{1}{2}g(\tilde{z})+\frac{1}{2}h(\tilde{z})+(h(\tilde{z})-h(\xo_i))\\
&=\frac{1}{2}\int_0^{\tilde{z}}\dvx{\F}{x_i}(x,\Xbf^{(i-1)}_{-i})dx+\frac{1}{2}\int_{\tilde{z}}^1-\dvx{\F}{x_i}(x,\Ybf^{(i-1)}_{-i})dx+\int_{\tilde{z}}^{\xo_i}-\dvx{\F}{x_i}(x,\Ybf^{(i-1)}_{-i})\\
&\overset{(1)}{\geq} \frac{\tilde{z}}{2}\cdot\dvx{\F}{x_i}(\tilde{z},\Xbf^{(i-1)}_{-i})+\frac{(1-\tilde{z})}{2}\cdot\left(-\dvx{\F}{x_i}(\tilde{z},\Ybf^{(i-1)}_{-i})\right)+(\xo_i-\tilde{z})\left(-\dvx{\F}{x_i}(\tilde{z},\Ybf^{(i-1)}_{-i})\right)\\
&=\frac{\tilde{z}\alpha}{2}+\frac{(1-\tilde{z})\beta}{2}+(\xo_i-\tilde{z})\beta\\
&\overset{(2)}{\geq} \frac{\tilde{z}\alpha}{2}+\frac{(1-\tilde{z})\beta}{2}-\tilde{z}\beta\\
&\overset{(3)}{=}\frac{\alpha^2}{2(\alpha+\beta)}+\frac{\beta^2}{2(\alpha+\beta)}-\frac{\alpha\beta}{(\alpha+\beta)}=\frac{(\alpha-\beta)^2}{2(\alpha+\beta)}\geq 0,
\end{align*} 
where inequality (1) holds due to the coordinate-wise concavity of $\F$, inequality (2) holds as $\beta\geq 0$ and $\xo_i\geq 0$, and equality (3) holds as $\beta\tilde{z}=\alpha(1-\tilde{z})$. 
\paragraph{Case 2 ($\mathbf{\tilde{z}< \xo_i}$):} This case is the reciprocal of Case 1, with a similar proof. Note that $g(\xo_i)-g(\tilde{z})\geq h(\xo_i)-h(\tilde{z})$ due to \texttt{strong DR-SM} and the fact that $\tilde{z}< \xo_i$, so:
\begin{align*}
\val_S(\tilde{z},\xo_i)&=\frac{1}{2}g(\tilde{z})+\frac{1}{2}h(\tilde{z})+(g(\tilde{z})-g(\xo_i))\\
&=\frac{1}{2}\int_0^{\tilde{z}}\dvx{\F}{x_i}(x,\Xbf^{(i-1)}_{-i})dx+\frac{1}{2}\int_{\tilde{z}}^1-\dvx{\F}{x_i}(x,\Ybf^{(i-1)}_{-i})dx+\int_{\xo_i}^{\tilde{z}}\dvx{\F}{x_i}(x,\Xbf^{(i-1)}_{-i})\\
&\overset{(1)}{\geq} \frac{\tilde{z}}{2}\cdot\dvx{\F}{x_i}(\tilde{z},\Xbf^{(i-1)}_{-i})+\frac{(1-\tilde{z})}{2}\cdot\left(-\dvx{\F}{x_i}(\tilde{z},\Ybf^{(i-1)}_{-i})\right)+(\tilde{z}-\xo_i)\left(\dvx{\F}{x_i}(\tilde{z},\Xbf^{(i-1)}_{-i})\right)\\
&=\frac{\tilde{z}\alpha}{2}+\frac{(1-\tilde{z})\beta}{2}+(\tilde{z}-\xo_i)\alpha\\
&\overset{(2)}{\geq} \frac{\tilde{z}\alpha}{2}+\frac{(1-\tilde{z})\beta}{2}+(\tilde{z}-1)\alpha\\
&\overset{(3)}{=}\frac{\alpha^2}{2(\alpha+\beta)}+\frac{\beta^2}{2(\alpha+\beta)}-\frac{\alpha\beta}{(\alpha+\beta)}=\frac{(\alpha-\beta)^2}{2(\alpha+\beta)}\geq 0,
\end{align*} 
where inequality (1) holds due to the coordinate-wise concavity of $\F$, inequality (2) holds as $\alpha\geq 0$ and $\xo_i\leq 1$, and equality (3) holds as $\beta\tilde{z}=\alpha(1-\tilde{z})$. 
\end{proof}
Combining \Cref{prop:main-strong} and \Cref{lem:reduction-to-zero-sum-strong} for $\delta=2C\epsilon$ finishes the analysis and the proof of \Cref{thm:alg-strong}.


\section{Experimental Results}

\label{sec:exp}
We empirically measure the solution quality of three algorithms: \Cref{alg:bigreedy} (\NRWGameAlg), \Cref{alg:binary-search} (\NRWBinarySearchAlg) and the Bi-Greedy algorithm of \cite{bian2017guaranteed} (\BMBKAlg). These are all based on a double-greedy framework, which we implemented to iterate over coordinates in a random order. These algorithms also do not solely rely on oracle access to the function; they invoke one-dimensional optimizers, concave envelopes, and derivatives. We implement the first and the second (\Cref{alg:apx-opt} and \Cref{alg:apx-ce} in the supplement), and numerically compute derivatives by discretization. We consider two application domains, namely Non-concave Quadratic Programming (NQP)~\citep{bian2017guaranteed,kim2003exact,luo2010semidefinite}, under both \texttt{strong-DR} and \texttt{weak-DR}, and maximization of softmax extension for MAP inference of determinantal point process\citep{kulesza2012determinantal, gillenwater2012near}. Each experiment consists of twenty repeated trials. For each experiment, we use $n = 100$ dimensional functions. Our experiments were implemented in python. See the supplementary materials for the detailed specifics of each experiment. The results of our experiments are  in \Cref{tab:average-qual}, and the corresponding box and whisker plots are in \Cref{fig:box-whisker}. The data suggests that for all three experiments the three algorithms obtain very similar objective values. For example, in the \texttt{weak-DR} NQP experiment, the upper and lower quartiles are distant by roughly $10$, while the mean values of the three algorithms deviate by less than $1$.


\begin{table}[h]

\centering
\leftskip-0.6cm
\begin{tabular}{|c|c|c|c|}
\hline
    & \small{NQP,~$\forall i,j: H_{i,j}\leq 0, $ (\texttt{strong-DR})} & \small{NQP,~$\forall i\neq j: H_{i,j}\leq 0, $ (\texttt{weak-DR})} & \small{Softmax Ext. (\texttt{strong-DR})}  \\\hhline{|=|=|=|=|}
\multicolumn{1}{||l||}{\NRWGameAlg }                  &$1225.840375$                                       &$1203.288522$  & $24.056934$ \\ \hline
\multicolumn{1}{||l||}{\NRWBinarySearchAlg}              &                                      $1225.816408$ &$1202.737999$  &$23.945428$  \\ \hline
\multicolumn{1}{||l||}{\BMBKAlg}                 &                                 $1225.738044$ &$1203.424957$  & $24.055435$. \\ \hline
\end{tabular}

\caption{\label{tab:average-qual}Average objective value of $T=20$ repeated trials, with dimension $n=100$}
\end{table}
\begin{figure}
        \centering
 
        \begin{subfigure}{0.5\textwidth}
        \centering
           
                \includegraphics[scale=0.6]{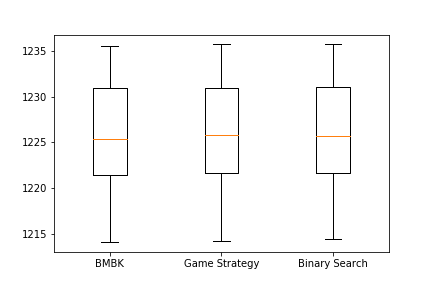}
\caption{\texttt{Strong DR-SM} NQP}
\label{fig:qp}
        \end{subfigure}
        \begin{subfigure}{0.5\textwidth}
        \centering
               \includegraphics[scale=0.6]{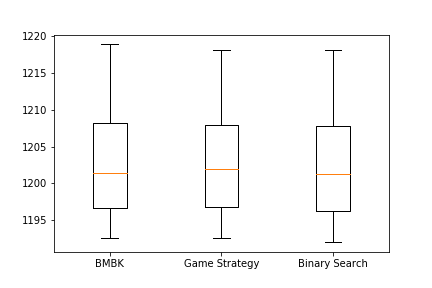}
\caption{\texttt{Weak DR-SM} NQP}
\label{fig:qpweak}
        \end{subfigure} \\
        \begin{subfigure}{0.5\textwidth}
        \centering
               \includegraphics[scale=0.6]{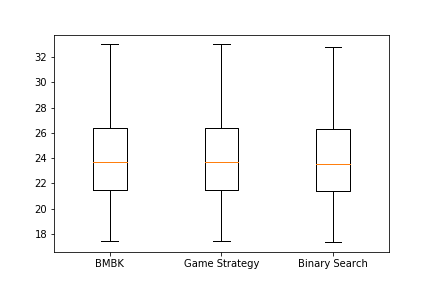}
\caption{\texttt{Strong DR-SM} Softmax}
\label{fig:softmax}
        \end{subfigure}
        \caption{\label{fig:box-whisker}Box and whisker plots of our experimental results.}
\end{figure}




\section{Conclusion}
We proposed a tight approximation algorithm for continuous submodular maximization, and a quasilinear time tight approximation algorithm for the special case of DR-submodular maxmization. Our experiments also verify the applicability of these algorithms in practical domains in machine learning. One interesting avenue for future research is to generalize our techniques to maximization over any arbitrary separable convex set, which would broaden the application domains.
\section{Acknowledgments}
Rad Niazadeh was supported by Stanford Motwani fellowship. The authors would also like to thank Jan Vondr\'ak for helpful comments and discussions on an earlier draft of this work. 
\bibliographystyle{plainnat}
\bibliography{refs}

\begin{thebibliography}{30}
\providecommand{\natexlab}[1]{#1}
\providecommand{\url}[1]{\texttt{#1}}
\expandafter\ifx\csname urlstyle\endcsname\relax
  \providecommand{\doi}[1]{doi: #1}\else
  \providecommand{\doi}{doi: \begingroup \urlstyle{rm}\Url}\fi

\bibitem[Antoniadis et~al.(2011)Antoniadis, Gijbels, and
  Nikolova]{antoniadis2011penalized}
Anestis Antoniadis, Ir{\`e}ne Gijbels, and Mila Nikolova.
\newblock Penalized likelihood regression for generalized linear models with
  non-quadratic penalties.
\newblock \emph{Annals of the Institute of Statistical Mathematics},
  63\penalty0 (3):\penalty0 585--615, 2011.

\bibitem[Bach et~al.(2013)]{bach2013learning}
Francis Bach et~al.
\newblock Learning with submodular functions: A convex optimization
  perspective.
\newblock \emph{Foundations and Trends{\textregistered} in Machine Learning},
  6\penalty0 (2-3):\penalty0 145--373, 2013.

\bibitem[Bian et~al.(2017{\natexlab{a}})Bian, Levy, Krause, and
  Buhmann]{bian2017continuous}
An~Bian, Kfir Levy, Andreas Krause, and Joachim~M Buhmann.
\newblock Continuous {DR}-submodular maximization: Structure and algorithms.
\newblock In \emph{Advances in Neural Information Processing Systems}, pages
  486--496, 2017{\natexlab{a}}.

\bibitem[Bian et~al.(2017{\natexlab{b}})Bian, Mirzasoleiman, Buhmann, and
  Krause]{bian2017guaranteed}
Andrew~An Bian, Baharan Mirzasoleiman, Joachim Buhmann, and Andreas Krause.
\newblock Guaranteed non-convex optimization: Submodular maximization over
  continuous domains.
\newblock In \emph{Artificial Intelligence and Statistics}, pages 111--120,
  2017{\natexlab{b}}.

\bibitem[Buchbinder and Feldman(2016)]{buchbinder2016deterministic}
Niv Buchbinder and Moran Feldman.
\newblock Deterministic algorithms for submodular maximization problems.
\newblock In \emph{Proceedings of the twenty-seventh annual ACM-SIAM symposium
  on Discrete algorithms}, pages 392--403. SIAM, 2016.

\bibitem[Buchbinder et~al.(2015)Buchbinder, Feldman, Seffi, and
  Schwartz]{buchbinder2015tight}
Niv Buchbinder, Moran Feldman, Joseph Seffi, and Roy Schwartz.
\newblock A tight linear time (1/2)-approximation for unconstrained submodular
  maximization.
\newblock \emph{SIAM Journal on Computing}, 44\penalty0 (5):\penalty0
  1384--1402, 2015.

\bibitem[Calinescu et~al.(2011)Calinescu, Chekuri, P{\'a}l, and
  Vondr{\'a}k]{calinescu2011maximizing}
Gruia Calinescu, Chandra Chekuri, Martin P{\'a}l, and Jan Vondr{\'a}k.
\newblock Maximizing a monotone submodular function subject to a matroid
  constraint.
\newblock \emph{SIAM Journal on Computing}, 40\penalty0 (6):\penalty0
  1740--1766, 2011.

\bibitem[Chen et~al.(2018)Chen, Hassani, and Karbasi]{chen2018online}
Lin Chen, Hamed Hassani, and Amin Karbasi.
\newblock Online continuous submodular maximization.
\newblock \emph{arXiv preprint arXiv:1802.06052}, 2018.

\bibitem[Djolonga and Krause(2014)]{djolonga2014map}
Josip Djolonga and Andreas Krause.
\newblock From map to marginals: Variational inference in bayesian submodular
  models.
\newblock In \emph{Advances in Neural Information Processing Systems}, pages
  244--252, 2014.

\bibitem[Feige et~al.(2011)Feige, Mirrokni, and Vondrak]{feige2011maximizing}
Uriel Feige, Vahab~S Mirrokni, and Jan Vondrak.
\newblock Maximizing non-monotone submodular functions.
\newblock \emph{SIAM Journal on Computing}, 40\penalty0 (4):\penalty0
  1133--1153, 2011.

\bibitem[Gillenwater et~al.(2012)Gillenwater, Kulesza, and
  Taskar]{gillenwater2012near}
Jennifer Gillenwater, Alex Kulesza, and Ben Taskar.
\newblock Near-optimal map inference for determinantal point processes.
\newblock In \emph{Advances in Neural Information Processing Systems}, pages
  2735--2743, 2012.

\bibitem[Gotovos et~al.(2015)Gotovos, Karbasi, and Krause]{gotovos2015non}
Alkis Gotovos, Amin Karbasi, and Andreas Krause.
\newblock Non-monotone adaptive submodular maximization.
\newblock In \emph{Twenty-Fourth International Joint Conference on Artificial
  Intelligence}, 2015.

\bibitem[Graham(1972)]{graham1972efficient}
Ronald~L Graham.
\newblock An efficient algorith for determining the convex hull of a finite
  planar set.
\newblock \emph{Information processing letters}, 1\penalty0 (4):\penalty0
  132--133, 1972.

\bibitem[Hartline et~al.(2008)Hartline, Mirrokni, and
  Sundararajan]{hartline2008optimal}
Jason Hartline, Vahab Mirrokni, and Mukund Sundararajan.
\newblock Optimal marketing strategies over social networks.
\newblock In \emph{Proceedings of the 17th international conference on World
  Wide Web}, pages 189--198. ACM, 2008.

\bibitem[Hassani et~al.(2017)Hassani, Soltanolkotabi, and
  Karbasi]{hassani2017gradient}
Hamed Hassani, Mahdi Soltanolkotabi, and Amin Karbasi.
\newblock Gradient methods for submodular maximization.
\newblock In \emph{Advances in Neural Information Processing Systems}, pages
  5843--5853, 2017.

\bibitem[Ito and Fujimaki(2016)]{ito2016large}
Shinji Ito and Ryohei Fujimaki.
\newblock Large-scale price optimization via network flow.
\newblock In \emph{Advances in Neural Information Processing Systems}, pages
  3855--3863, 2016.

\bibitem[Iwata et~al.(2001)Iwata, Fleischer, and
  Fujishige]{iwata2001combinatorial}
Satoru Iwata, Lisa Fleischer, and Satoru Fujishige.
\newblock A combinatorial strongly polynomial algorithm for minimizing
  submodular functions.
\newblock \emph{Journal of the ACM (JACM)}, 48\penalty0 (4):\penalty0 761--777,
  2001.

\bibitem[Kapralov et~al.(2013)Kapralov, Post, and
  Vondr{\'a}k]{kapralov2013online}
Michael Kapralov, Ian Post, and Jan Vondr{\'a}k.
\newblock Online submodular welfare maximization: Greedy is optimal.
\newblock In \emph{Proceedings of the twenty-fourth annual ACM-SIAM symposium
  on Discrete algorithms}, pages 1216--1225. SIAM, 2013.

\bibitem[Kim and Kojima(2003)]{kim2003exact}
Sunyoung Kim and Masakazu Kojima.
\newblock Exact solutions of some nonconvex quadratic optimization problems via
  sdp and socp relaxations.
\newblock \emph{Computational Optimization and Applications}, 26\penalty0
  (2):\penalty0 143--154, 2003.

\bibitem[Krause and Golovin(2014)]{krause2014submodular}
Andreas Krause and Daniel Golovin.
\newblock Submodular function maximization.
\newblock In \emph{Tractability: Practical Approaches to Hard Problems}, pages
  71--104. Cambridge University Press, 2014.

\bibitem[Kulesza et~al.(2012)Kulesza, Taskar, et~al.]{kulesza2012determinantal}
Alex Kulesza, Ben Taskar, et~al.
\newblock Determinantal point processes for machine learning.
\newblock \emph{Foundations and Trends{\textregistered} in Machine Learning},
  5\penalty0 (2--3):\penalty0 123--286, 2012.

\bibitem[Li et~al.(2016)Li, Sra, and Jegelka]{li2016fast}
Chengtao Li, Suvrit Sra, and Stefanie Jegelka.
\newblock Fast mixing markov chains for strongly rayleigh measures, dpps, and
  constrained sampling.
\newblock In \emph{Advances in Neural Information Processing Systems}, pages
  4188--4196, 2016.

\bibitem[Luo et~al.(2010)Luo, Ma, So, Ye, and Zhang]{luo2010semidefinite}
Zhi-Quan Luo, Wing-Kin Ma, Anthony Man-Cho So, Yinyu Ye, and Shuzhong Zhang.
\newblock Semidefinite relaxation of quadratic optimization problems.
\newblock \emph{IEEE Signal Processing Magazine}, 27\penalty0 (3):\penalty0
  20--34, 2010.

\bibitem[Mirzasoleiman et~al.(2013)Mirzasoleiman, Karbasi, Sarkar, and
  Krause]{mirzasoleiman2013distributed}
Baharan Mirzasoleiman, Amin Karbasi, Rik Sarkar, and Andreas Krause.
\newblock Distributed submodular maximization: Identifying representative
  elements in massive data.
\newblock In \emph{Advances in Neural Information Processing Systems}, pages
  2049--2057, 2013.

\bibitem[Roughgarden and Wang(2018)]{roughgarden2018optimal}
Tim Roughgarden and Joshua~R Wang.
\newblock An optimal learning algorithm for online unconstrained submodular
  maximization.
\newblock In \emph{To Appear in Proceedings of the 31st Conference on Learning
  Theory (COLT)}, 2018.

\bibitem[Schrijver(2000)]{schrijver2000combinatorial}
Alexander Schrijver.
\newblock A combinatorial algorithm minimizing submodular functions in strongly
  polynomial time.
\newblock \emph{Journal of Combinatorial Theory, Series B}, 80\penalty0
  (2):\penalty0 346--355, 2000.

\bibitem[Soma and Yoshida(2015)]{soma2015generalization}
Tasuku Soma and Yuichi Yoshida.
\newblock A generalization of submodular cover via the diminishing return
  property on the integer lattice.
\newblock In \emph{Advances in Neural Information Processing Systems}, pages
  847--855, 2015.

\bibitem[Soma and Yoshida(2017)]{soma2017non}
Tasuku Soma and Yuichi Yoshida.
\newblock Non-monotone dr-submodular function maximization.
\newblock In \emph{AAAI}, volume~17, pages 898--904, 2017.

\bibitem[Staib and Jegelka(2017)]{staib2017robust}
Matthew Staib and Stefanie Jegelka.
\newblock Robust budget allocation via continuous submodular functions.
\newblock \emph{arXiv preprint arXiv:1702.08791}, 2017.

\bibitem[Zhang et~al.(2015)Zhang, Djolonga, and Krause]{zhang2015higher}
Jian Zhang, Josip Djolonga, and Andreas Krause.
\newblock Higher-order inference for multi-class log-supermodular models.
\newblock In \emph{Proceedings of the IEEE International Conference on Computer
  Vision}, pages 1859--1867, 2015.

\end{thebibliography}
\appendix
\section*{Supplementary Materials}

\section*{Equivalent definitions of weakly and strongly DR-SM functions.}
\begin{proposition}[\citep{bian2017guaranteed}] 
\label{prop:weak-strong-multi-defs}
Suppose $\F:[0,1]^n\rightarrow [0,1]$ is continuous and twice differentiable, and $\Hbf$ is the Hessian of $\F$, i.e. $\forall i, j\in[n],~H_{ij}\triangleq \dvxy{\F}{x_i}{x_j}$. The followings are equivalent: 
\begin{enumerate}
\item $\F$ satisfies the \texttt{\emph{weak DR-SM}} property as in \Cref{def:weak-strong-DR}.
\item Continuous submodularity: $\forall \xv,\yv\in[0,1]^n$, $\F(\xv)+\F(\yv)\geq \F(\xv\vee\yv)+\F(\xv\wedge\yv)$.
\item $\forall i\neq j\in[n],~H_{ij}\leq 0$, \ie all off-diagonal entries of Hessian are non-positive.
\end{enumerate}
Also, the following statements are equivalent:
\begin{enumerate}
\item$\F$ satisfies the \texttt{\emph{strong DR-SM}} property as in \Cref{def:weak-strong-DR}.
\item $\F(.)$ is coordinate-wise concave along all the coordinates and is continuous submodular, i.e. $\forall \xv,\yv\in[0,1]^n$, $\F(\xv)+\F(\yv)\geq \F(\xv\vee\yv)+\F(\xv\wedge\yv)$ 
\item $\forall i, j\in[n],~H_{ij}\leq 0$, \ie all entries of Hessian are non-positive.
\end{enumerate}
\end{proposition}

\section*{Detailed specifics of experiments in \Cref{sec:exp}}
\subsection*{\texttt{Strong-DR} Non-concave Quadratic Programming (NQP)}
 We generated synthetic functions of the form $\F(\xbf) = \frac12 \xbf^T \Hbf \xbf + \hbf^T \xbf + c$. We generated $\Hbf \in \RR^{n \times n}$ as a matrix with every entry \emph{uniformly distributed} in $[-1, 0]$, and then symmetrized $\Hbf$. We then generated $\hbf \in \RR^{n}$ as a vector with every entry uniformly distributed in $[0, +1]$. Finally, we solved for the value of $c$ to make $\F(\vec{0}) + \F(\vec{1}) = 0$.
\subsection*{\texttt{Weak-DR} Non-concave Quadratic Programming (NQP)}

This experiment is the same as in the previous subsection, except that the diagonal entries of $\Hbf$ are uniformly distributed in $[0, +1]$ instead, making the resulting function $\F(\xbf)$ only \texttt{weak DR-SM} instead.
\subsection*{Softmax extension of Determinantal Point Processes (DPP)}

We generated synthetic functions of the form $\F(\xbf) = \log \det (\text{diag}(\xbf) (\Lbf - \Ibf) + \Ibf)$, where $\Lbf$ needs to be positive semidefinite. We generated $\Lbf$ in the following way. First, we generate each of the $n$ eigenvalues by drawing a uniformly random number in $[-0.5, 1.0]$ and taking that power of $e$. This yields a diagonal matrix $\Dbf$. We then generate a random unitary matrix $\Vbf$ and then set $\Lbf = \Vbf \Dbf \Vbf^T$. By construction, $\Lbf$ is positive semidefinite and has the specified eigenvalues.


\section*{More application domain details}

Here is a list containing further details about applications in machine learning, electrical engineering and other application domains.
\paragraph{Special Class of Non-Concave Quadratic Programming (NQP).}
\begin{itemize}
\item The objective is to maximize $\F(\xbf)=\frac{1}{2}\xbf^T\Hbf\xbf+\hbf^T\xbf+c$, where off-diagonal entries of $\Hbf$ are non-positive (and hence these functions are \texttt{Weak DR-SM}).
\item Minimization of this function (or equivalently maximization of this function when off-diagonal entries of $\Hbf$ are non-negative) have been studied in \citet{kim2003exact} and \citet{luo2010semidefinite}, and has applications in communication systems and detection in MIMO channels~\citep{luo2010semidefinite}.
\item Another application of quadratic submodular optimization is large-scale price optimization on the basis of demand forecasting models, which has been studied in~\cite{ito2016large}. They show the price optimization problem is indeed an instance of weak-DR minimization.
\end{itemize}
\paragraph{Revenue Maximization over Social Networks.}
\begin{itemize}
\item The model was proposed in~\cite{bian2017guaranteed} and is a generalization of the revenue maximization problem addressed in \cite{hartline2008optimal}.
\item A seller wishes to sell a product to a social network of buyers. We consider restricted seller strategies which freely give (possibly fractional) trial products to buyers: this fractional assignment is our input $\xbf$ of interest.
\item The objective takes two effects into account: (i) the revenue gain from buyers who didn't receive free product, where the revenue function for each such buyer is a nonnegative nondecreasing \texttt{Weak DR-SM} function and (ii) the revenue loss from those who received free product, where the revenue function for each such buyer is a nonpositive nonincreasing \texttt{Weak DR-SM} function. The combination for all buyers is a nonmonotone \texttt{Weak DR-SM} function and additionally is nonnegative at $\vec{0}$ and $\vec{1}$.
\end{itemize}
\paragraph{Map Inference for Determinantal Point Processes (DPP) \& Its Softmax-Extension.}
\begin{itemize}
\item DPP are probabilistic models that arise in statistical physics and random matrix theory, and their applications in machine learning have been recently explored, e.g.~\citep{kulesza2012determinantal}.
\item DPPs can be used as generative models in applications such as text summarization, human pose estimation, or news threading tasks~\citep{kulesza2012determinantal}.
\item A discrete DPP is a distribution over sets, where $p(S)\sim \det(A_S)$ for a given PSD matrix $A$. The log-likelihood estimation task corresponds to picking a set $\hat{S}\in\mathcal{P}$(feasible set, e.g. a matching) that maximizes $f(S)=\log(\det(A_S))$. This function is non-monotone and submodular. Note that as a technical condition to apply bi-greedy algorithms, we require that $\det(A)\geq 1$ (implying $f(\vec{1})\geq 0$).
\item The approximation question was studied in \citep{gillenwater2012near}. Their idea is to first find a fractional solution for a continuous extension (hence a a continuous submodular maximization step is required) and then rounding the solution. However, they sometimes need a fractional solution in $\textrm{conv}(\mathcal{P})$ (so, the optimization task sometimes fall out of the hypercube, making rounding more complicated). 
\item Beyond multilinear extension, the other continuous extension that has been used in this literature is called the \emph{softmax extension}~\citep{gillenwater2012near,bian2017continuous}: $$\F(\xbf)=\log\mathbf{E}_{S\sim\mathcal{I}_\xbf}[\exp(f(S))]=\log\det\left(\textrm{diag}(\xbf)(A-I)+I\right)$$
where $\mathcal{I}_{\xbf}$ is the independent distribution with marginals $\xbf$ (i.e. each item $i$ is independently in the set w.p. $x_i$). 
\item $\F(\xbf)$ is \texttt{Strong DR-SM} and non-monotone~\citep{bian2017continuous}. In almost all machine learning applications, the rounding works on an unrestricted problem. Hence the optimization that needs to be done is \texttt{Strong DR-SM} optimization over unit hypercube.
\item One can think of adding a regularizer term $\lambda \lVert\xbf \rVert^2$ to the log-likelihood objective function to avoid overfitting. In that case, the underlying fractional problem becomes a \texttt{Weak DR-SM} optimization over the unit hypercube when $\lambda$ is large enough.
\end{itemize}
\paragraph{Log-Submodularity and Mean-Field Inference.}
\begin{itemize}
\item Another probabilistic model that generalizes DPP and all other strong Rayleigh measures~\citep{li2016fast,zhang2015higher} is the class of \emph{log-submodular} distributions over sets, i.e. $p(S)\sim\exp(f(S))$ where $f(\cdot)$ is a discrete submodular functions. MAP inference over this distribution has applications in machine learning and beyond~\citep{djolonga2014map}.
\item One variational approach towards this MAP inference task is to do \emph{mean-field inference} to approximate the distribution $p$ with a product distribution $\xbf\in [0,1]^n$, i.e. finding  $\xbf^*$ that:
$$\xbf^*\in\underset{\xbf\in[0,1]^n}{\argmax}~~\mathbb{H}(\xbf)-\mathbf{E}_{S\sim \mathcal{I}_x}[\log p(S)]=\underset{\xbf\in[0,1]^n}{\argmin}~~\textrm{KL}(\xbf||p)$$
where $\textrm{KL}(\xbf||p)=\mathbf{E}_{S\sim \mathcal{I}}[\frac{\log \mathcal{I}_x(S)}{\log p(S)}]$.
\item The function $\F(\xbf)=\mathbb{H}(\xbf)-\mathbf{E}_{S\sim \mathcal{I}_x}[\log p(S)]$ is \texttt{Strong DR-SM}~\citep{bian2017continuous}.
\end{itemize}
\paragraph{Cone Extension of Continuous Submodular Maximization.}
\begin{itemize}
\item Suppose $\mathcal{K}$ is a proper cone. By considering the lattice corresponding to this cone one can generalize DR submodularity to $\mathcal{K}$-DR submodularity~\citep{bian2017continuous}.
\item An interesting application of this cone generalization is minimizing the loss in the logistic regression model with a particular non-separable and non-convex regularizer, as described in~\citep{antoniadis2011penalized,bian2017continuous}. \citet{bian2017continuous} show the vanilla version is a $\mathcal{K}$-\texttt{Strong DR-SM} function maximization for some particular cone.
\item Note that by adding a $\mathcal{K}$-$\ell_2$-norm regularizer $\lambda\lVert \textbf{A}\xbf\rVert^2$, the function will become \texttt{Weak DR-SM}, where $\textbf{A}$ is a matrix with generators of $\mathcal{K}$ as its column. Here is the logistic loss:
$$ l(\xbf,\{y_t\})=\frac{1}{T}\sum_{t=1}^T f_t(\xbf,y_t)=\frac{1}{T}\sum_{t=1}^T \log(1+\exp(-y_t\xbf^T\mathbf{z}^t))$$
where $y_t$ is the label of the $t^{\textrm{th}}$ data-point, $\xbf$ are the model parameters, and $\{\mathbf{z}^t\}$ are feature vectors of the data-points.
\end{itemize}
\begin{remark}
In many machine learning applications, and in particular MAP inference of DPPs and log-submodular distributions, unless we impose some technical assumptions, the underlying \texttt{Strong DR-SM} (or \texttt{Weak DR-SM}) function is \emph{not} necessarily positive (or may not even satisfy the weaker yet sufficient condition $\F(\vec{0})+\F(\vec{1})\geq 0$). In those cases, adding a positive constant to the function can fix the issue, but the multiplicative approximation guarantee becomes weaker. However, this trick tends to work in practice since these algorithms tend to be near optimal.
\end{remark}

\end{document}